\documentclass[12pt]{article}
\usepackage[margin=1in]{geometry}
\usepackage{amsmath,amsthm,amssymb}
\usepackage{setspace}
\usepackage{natbib}
\usepackage{bm}
\usepackage{mathrsfs}
\usepackage{comment}
\usepackage{booktabs}

\usepackage{tikz}
\usetikzlibrary{arrows.meta,positioning}

\newcommand{\Bernoulli}{\operatorname{Bernoulli}}
\newcommand{\Beta}{\operatorname{Beta}}
\newcommand{\Binomial}{\operatorname{Binomial}}
\newcommand{\Branches}{\mathcal{B}}
\newcommand{\cloglog}{\texttt{cloglog}}

\newcommand{\Data}{\mathcal D}
\newcommand{\DP}{\operatorname{DP}}
\newcommand{\Dirichlet}{\operatorname{Dirichlet}}
\newcommand{\E}{\mathbb{E}}
\newcommand{\Exp}{\operatorname{Exp}}
\newcommand{\Gam}{\operatorname{Gam}}

\newcommand{\logit}{\operatorname{logit}}
\newcommand{\Leaves}{\mathcal{L}}
\newcommand{\sM}{\mathcal M}
\newcommand{\sQ}{\mathcal Q}
\newcommand{\Normal}{\operatorname{Normal}}

\newcommand{\TExp}{\operatorname{TExp}}
\newcommand{\Tree}{\mathcal T}
\newcommand{\Uniform}{\operatorname{Uniform}}
\newcommand{\Var}{\operatorname{Var}}
\newcommand{\R}{\mathbb R}

\newcommand{\scriptX}{\mathscr X}

\newtheorem{theorem}{Theorem}

\newtheorem{lemma}{Lemma}
\theoremstyle{definition}
\newtheorem{remark}{Remark}

\usepackage{xcolor}  

\usepackage[colorlinks, citecolor = blue]{hyperref}
\usepackage{xr}
\begin{document}

\title{A Unified Bayesian Nonparametric Framework for Ordinal, Survival, and Density
  Regression Using the Complementary Log-Log Link}
\date{}

\author{Entejar Alam\thanks{\texttt{entejar@utexas.edu}} \ and
  Antonio R. Linero\thanks{\texttt{antonio.linero@austin.utexas.edu}}}

\maketitle


\begin{abstract}
  In this work, we develop applications of the complementary log-log (\cloglog) link to problems in Bayesian nonparametrics. Although less commonly used than the probit or logit links, we find that the \cloglog\ link is computationally and theoretically well-suited to several commonly used Bayesian nonparametric methods. Our starting point is a Bayesian nonparametric model for ordinal regression. We first review how the \cloglog\ link uniquely sits at the intersection of the \emph{cumulative link} and \emph{continuation ratio} approaches to ordinal regression. Then, we develop a convenient computational method for fitting these ordinal models using Bayesian additive regression trees. Next, we use our ordinal regression model to build a Bayesian nonparametric stick-breaking process and show that, under a proportional hazards assumption, our stick-breaking process can be used to construct a weight-dependent Dirichlet process mixture model. Again, Bayesian additive regression trees lead to convenient computations. We then extend these models to allow for Bayesian nonparametric survival analysis in both discrete and continuous time. Our models have desirable theoretical properties, and we illustrate this analyzing the posterior contraction rate of our ordinal models. Finally, we demonstrate the practical utility of our \cloglog\ models through a series of illustrative examples.

  \vspace{1em}

\end{abstract}

\doublespacing

\section{Introduction}

Logistic and probit regression models have been important tools in Bayesian modeling \citep[][Chapter 16]{gelman2013bayesian}, both for modeling binary data and as building blocks for more complicated models. Examples of interesting Bayesian nonparametric applications of these models include the logistic \citep{ren2011logistic} and probit \citep{rodriguez2011nonparametric} stick-breaking processes, discrete (or continuous) time hazard models \citep{sparapani2016nonparametric}, and ordinal regression models based on either cumulative links \citep{albert1993bayesian} or continuation-ratios \citep[][Chapter 8]{agresti2013categorical}. The standard probit and logit models for binomial data set $Y_i \sim \Binomial(n_i, \mu_i)$ with $g(\mu_i) = X_i^\top \beta$, where $g(\mu) = \logit(p)$ for the logistic model and $g(\mu) = \Phi^{-1}(p)$ for the probit model.

Despite its potential advantages, the complementary log-log (\cloglog) link function $g(\mu) = \log\{-\log(1 - \mu)\}$ has received less attention compared to the more commonly used probit and logit links. Rather, the \cloglog\ link has typically been applied in relatively narrow, subject-specific contexts. The \cloglog\ link is particularly useful when an \emph{asymmetric} link function is required, as it captures skewed relationships between predictors and outcomes. For instance, in toxicological studies, increasing the dosage of a toxin beyond a certain threshold may lead to a rapid increase in the risk of death, while reducing an already low dose may have minimal impact on the risk.

One possible reason for the limited popularity of the \cloglog\ link among Bayesian practitioners is the lack of convenient data augmentation schemes. While significant progress has been made in developing effective data augmentation strategies for logistic and probit regression models \citep{polson2013bayesian}, similar advancements for the \cloglog\ link are lacking. Although a latent variable representation involving Gumbel distributions is known \citep{agresti2013categorical}, augmenting data with Gumbel latent variables does not straightforwardly simplify computations.

In this work, we develop several Bayesian nonparametric models that exploit the unique properties of the \cloglog\ link. By doing so, we construct methods that are computationally simpler and more robust than analogous approaches using the probit or logit links. Specifically, we consider the following applications:
\begin{itemize}

\item \textbf{Ordinal regression models:} Fitting cumulative link ordinal regression models with the \cloglog\ link entirely avoids the need to impose constraints on the breakpoints of the model. All updates are either conditionally conjugate or based on standard updates for Bayesian additive regression trees (BART) parameters. In contrast, probit or logit-based methods must contend with complex dependencies in the posterior distribution of the cutpoints \citep{cowles1996accelerating}.

\item \textbf{Stick-breaking processes:}  We introduce the \emph{proportional hazards stick-breaking process} (PHSBP) and the \emph{non-proportional hazards stick-breaking process} (NPHSBP) as priors for families of random probability measures $\{F_x : x \in \mathcal X\}$. A key advantage of these new processes relative to, for example, the probit stick breaking process \citep{rodriguez2011nonparametric} is that we can fit mixture models based these stick-breaking processes by augmenting only a single latent variable per individual.

\item \textbf{Survival Models:} The same computational advantages apply to discrete-time survival models with a proportional hazards structure and to continuous-time proportional hazards models with a piecewise-constant hazard function. We illustrate this by developing both proportional hazards and non-proportional hazards survival models.

\end{itemize}
To obtain these results, we leverage a novel truncated-exponential data augmentation scheme that integrates uniquely well with Bayesian additive regression trees \citep[BART,][]{chipman2010bart} methods. The use of BART is critical, as other nonparametric priors do not lead to conditional conjugacy in this context.

We demonstrate the utility of these methods through a variety of simulated and real data analyses. Our ordinal data models are evaluated using data from the Medical Expenditure Panel Survey (MEPS), where we show that, in addition to being easier to implement, survey questions pertaining to depression are better modeled with the \cloglog\ link than with the probit or logit links. Using the same data source, we employ an NPHSBP mixture model to replicate the analysis from \citet{li2023adaptive}, examining how the density function of body mass index varies across subpopulations based on education levels. Finally, we apply a proportional hazards survival model to data described by \citet{henderson2002modeling} to infer the effects of age and material deprivation (measured by the Townsend index) on survival from leukemia.

We also briefly describe the theoretical properties of our models. Using the
relationship between the \cloglog\ link and the Gumbel distribution, we are able to derive posterior concentration rates for our ordinal regression model; similar results are expected to hold for the other models we consider. These rates are minimax-optimal up to logarithmic factors. We also show that a special case of the PHSBP corresponds to a weight-dependent Dirichlet process \citep[DDP,][]{maceachern2000dependent}.

The remainder of the paper is organized as follows. In Section~\ref{sec:BART}, we provide an overview of the Bayesian Additive Regression Trees (BART) framework. In Section~\ref{sec:PH}, we develop the methods underlying our use of the \cloglog\ link, providing both theoretical insights and practical implementation details. We also describe non-proportional extensions of our models in Section~\ref{sec:PH}, which essentially correspond to fully nonparametric models, and discuss some of the properties of the PHSBP and NPHSBP models. Section~\ref{sec:posterior-concentration} establishes posterior concentration rates for our ordinal model. In Section~\ref{sec:illustrations}, we empirically evaluate our methods using both synthetic and real datasets. We conclude with a discussion in Section~\ref{sec:discussion}.


\section{Review of Bayesian Additive Regression Trees}
\label{sec:BART}

Before introducing our methodology, we provide a brief overview of the Bayesian Additive Regression Trees (BART) approach to nonparametric modeling. Introduced by \citet{chipman2010bart}, BART is a Bayesian nonparametric framework that combines the flexibility of decision tree ensembles with the rigorous probabilistic foundation of Bayesian statistics. In BART, a function of interest, $r(x)$, is expressed as a sum of \emph{regression trees}:
\[
  r(x) = \sum_{t=1}^T g(x; \Tree_t, \sM_t),
\]
where \( \Tree_t \) represents the structure of the \( t^{\text{th}} \) decision tree, comprising branch nodes (\( \Branches(\Tree_t) \)) and leaf nodes (\( \Leaves(\Tree_t) \)), and \( \sM_t = \{\mu_{t\ell} : \ell \in \Leaves(\Tree_t)\} \) denotes the associated \emph{leaf node} parameters. The function \( g(x; \Tree_t, \sM_t) \) is a step function that assigns the value \( \mu_{t\ell} \) to \( x \) if it falls within leaf \( \ell \) of \( \Tree_t \). Figure~\ref{fig:tree} provides a schematic illustration of a regression tree \( g(x; \Tree_t, \sM_t) \), depicted both as a decision tree and as a corresponding step function.

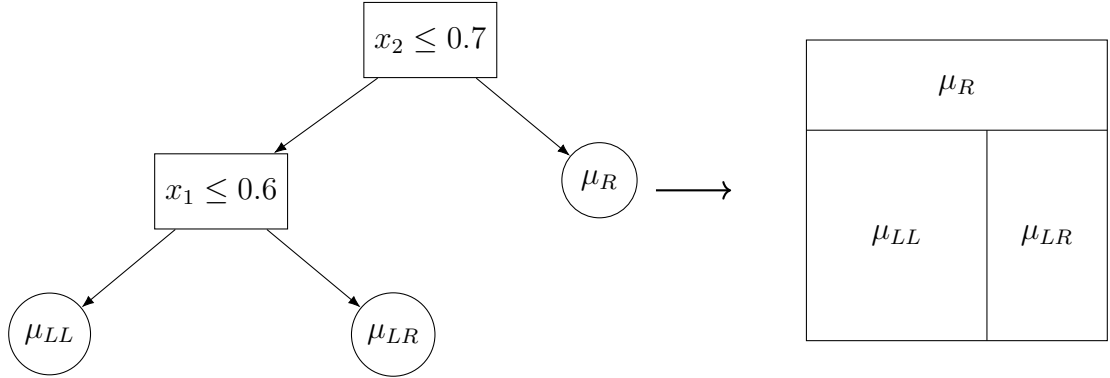
\begin{figure}
  \centering
    \begin{tikzpicture}
      \tikzset{every node/.style={draw,rectangle,minimum size=1cm}}
      \node (root) {$x_2 \leq 0.7$};
      \node[below left=of root] (left) {$x_1 \leq 0.6$};
      \node[below right=of root,circle] (right) {$\mu_R$};
      \node[below left=of left,circle] (leftleft) {$\mu_{LL}$};
      \node[below right=of left,circle] (leftright) {$\mu_{LR}$};

      \draw[-{Latex}] (root) -- (left);
      \draw[-{Latex}] (root) -- (right);
      \draw[-{Latex}] (left) -- (leftleft);
      \draw[-{Latex}] (left) -- (leftright);

      \draw[->, thick] (3,-2) -- (4,-2);

      \begin{scope}[shift={(5,-4)}]
          \draw (0,0) rectangle (4,4);
          \draw (0,2.8) -- (4,2.8);
          \draw (2.4,0) -- (2.4,2.8);
          \node[draw = none] at (2, 3.4) {$\mu_R$};
          \node[draw = none] at (1.2, 1.4) {$\mu_{LL}$};
          \node[draw = none] at (3.2, 1.4) {$\mu_{LR}$};
      \end{scope}

  \end{tikzpicture}
  \caption{Schematic of a regression tree in terms of a decision tree (left) and a step function on the predictor space $[0,1]^2$ (right). The tree splits the predictor space based on rules \(x_2 \leq 0.7\) and \(x_1 \leq 0.6\), resulting in three regions with terminal node values \(\mu_R\), \(\mu_{LL}\), and \(\mu_{LR}\)}
  \label{fig:tree}
\end{figure}


BART builds on earlier tree-based methods, such as random forests \citep{breiman2001random} and gradient boosting \citep{friedman2001greedy}, by placing a prior distribution, \( \pi_\Tree(\Tree) \, \pi_\sM(\sM_t \mid \Tree_t) \), on both the tree structures and the leaf node predictions. These priors are specifically designed to regularize the estimation process and mitigate overfitting. The decision trees in the ensemble are assumed to be a-priori independent and are governed by three components:  (i) a \emph{branching process} prior on the on the depth and shape of the trees that penalizes tree complexity; (ii) a prior on how the decision trees are constructed, where predictor $j$ is selected with some probability $s_j$ and then a splitting location is sampled uniformly; and (iii) a prior $\pi_\mu$ on the leaf node parameters, which is typically chosen to have prior mean $0$ and variance $\sigma^2_\mu \propto 1/T$. The scaling of the variance $\sigma^2_\mu$ ensures that individual trees contribute modestly to the overall fit. These priors collectively allow BART to balance flexibility with regularization, making it robust and effective for nonparametric regression tasks.

\begin{paragraph}{BART for Nonparametric Regression}
  \citet{chipman2010bart} originally proposed BART as a method for nonparametric regression, where the model assumes \( Y_i \sim \Normal\{r(X_i), \sigma^2\} \). They developed a \emph{Bayesian backfitting} algorithm to perform inference using Markov chain Monte Carlo (MCMC). This algorithm operates by sequentially sampling \( \Tree_t \) from the conditional distribution \( [\Tree_t \mid \Tree_{-t}, \sM_{-t}, \Data] \) using a Metropolis-Hastings step, followed by sampling \( \sM_t \) from \( [\sM_t \mid \Tree_t, \Tree_{-t}, \sM_{-t}, \Data] \). The implementation of these sampling steps depends heavily on the specification of a conjugate prior \( \pi_\mu(\mu) \), which facilitates efficient computation. However, it is also possible to use non-conjugate priors by employing reversible jump algorithms, as described by \citet{linero2024generalized}.
\end{paragraph}

\begin{paragraph}{BART for Non-Normal Likelihoods via Data Augmentation}
  BART has been extended to accommodate non-normal likelihoods in various ways. For instance, \citet{chipman2010bart} proposed a probit extension for classification tasks, using the data augmentation scheme of \citet{albert1993bayesian}. Of particular relevance to this work are the survival analysis models developed by \citet{linero2022bayesian} and \citet{sparapani2016nonparametric}, which demonstrate how BART can be applied to model time-to-event data both with and without Cox's proportional hazards assumption. \citet{sparapani2016nonparametric} employ a discrete-time survival model of the form
  \[
    \Pr(T_i = t \mid T_i \ge t, X_i = x) = \Phi\{r(t, x)\}, \qquad \text{for } t \in \{t_1, \ldots, t_K\},
  \]
  where \( r(t, x) \) is assigned a BART prior. Inference for this model can proceed using the data augmentation techniques described by \citet{albert1993bayesian}, although the efficiency of this approach is limited by the need to augment \( k \) latent variables whenever \( T_i = t_k \). When this framework is adapted for ordinal data rather than survival data, such models are referred to as \emph{continuation ratio models} \citep{agresti2013categorical}. BART with data augmentation has also been applied to fit \emph{cumulative probit} models for ordinal data, which specify
  \[
    \Pr(Y_i = k \mid X_i = x) = \Phi\{c_k - r(x)\} - \Phi\{c_{k-1} - r(x)\},
  \]
  where \( -\infty = c_0 < c_1 < \cdots < c_{K-1} < c_K = \infty \) \citep{albert1993bayesian, muller2007spatially}. While cumulative probit models impose more restrictive assumptions than continuation ratio models, they are sometimes preferred because of their simpler interpretation in terms of latent continuous variables \( Z_i \), where \( Z_i \sim \Normal\{r(X_i), 1\} \) and \( Y_i = k \) if \( Z_i \in [c_{k-1}, c_k) \). However, fitting cumulative probit models is more challenging due to the requirement that the thresholds \( c_k \) be constrained to increase monotonically \citep{cowles1996accelerating}.
\end{paragraph}

\begin{paragraph}{BART for Non-Normal Likelihoods Without Data Augmentation}
  Interestingly, simple Gibbs samplers are also available non-normal likelihoods in certain situations that completely avoid latent variables by simply modifying the leaf node prior $\pi_\mu(\mu)$ \citep{murray2021log}. A useful trick is to use a \emph{log-gamma} prior $\log \Gam(a, b)$, which allows us to fit proportional hazards regression models with Cox's partial likelihood \citep{linero2022bayesian}, gamma regression models with a fixed shape parameter \citep{linero2020semiparametric}, and Poisson process models \citep{lamprinakou2023bart}, among others. The values of $a$ and $b$ are chosen so that $\E(\mu_{t\ell}) = 0$ and $\Var(\mu_{t\ell}) = \sigma^2_\mu$, which can be shown to imply that $\psi(a) = \log b$ and $\psi'(a) = \sigma^2_\mu$, while $\psi(x) = \frac{d}{dx} \log \Gamma(x)$ is the digamma function. In this work, we set $\sigma_\mu = 1.5 / \sqrt T$ to mimic the prior specification for the original BART prior.
\end{paragraph}

\section{Inference for Complementary Log-Log Models}
\label{sec:PH}

The models we consider are unified in the sense that they all have representations in terms of latent exponential (or piecewise-exponential) random variables. Accordingly, we begin with a general description of BART inference with an exponential outcome
\begin{math}
  [E_i \mid X_i = x] \sim \Exp(e^{r(x)}),
\end{math}
and explain how inference for this model can be used to fit models that use the \cloglog\ link. Letting $Z_i = \log E_i$, we equivalently have the regression model $Z_i = -r(X_i) + \epsilon_i$ where $\epsilon_i$ follows a Gumbel distribution with density $h(\epsilon) = \exp(\epsilon - e^\epsilon)$.

Interestingly, the model for $Z_i$ is amenable to exactly the same Bayesian backfitting algorithm introduced by \citet{chipman2010bart}, provided that we use a \emph{log-gamma} prior $\mu_{t\ell} \sim \log \Gam(a, b)$ rather than the usual normal distribution $\mu_{t\ell} \sim \Normal(0, \sigma^2_\mu)$. Specifically, if we define the \emph{partial residuals} associated to $(\Tree_t, \sM_t)$ as
\begin{math}
  R_{it} = Z_i + \sum_{j \ne t} g(X_i; \Tree_j, \sM_j)
\end{math}
then the conditional likelihood of $(\Tree_t, \sM_t)$ given all the other trees can be written as
\begin{align*}
  L(\Tree_t, \sM_t) =
  \prod_{\ell \in \Leaves(\Tree_t)} \prod_{i: X_i \leadsto \ell} \exp\left( R_{it} + \mu_{t\ell} - e^{\mu_{t\ell}} e^{R_{it}} \right)
  \propto \prod_{\ell \in \Leaves(\Tree_t)}
  \exp(A_{t\ell} \mu_{t\ell} - B_{t\ell} e^{\mu_{t\ell}}),
\end{align*}
where $[X_i \leadsto \ell]$ denotes that $X_i$ is associated to leaf node $\ell$ of $\Tree_t$, $A_{t\ell} = \sum_{i: X_i \leadsto \ell} 1$, and $B_{t\ell} = \sum_{i: X_i \leadsto \ell} e^{R_{it}}$. Under the prior $\mu_{t\ell} \sim \log \Gam(a, b)$, we can therefore leverage conjugacy to compute the \emph{integrated likelihood}
\begin{align*}
  L(\Tree_t) = \int L(\Tree_t, \sM_t) \, \prod_{\ell \in \Leaves(\Tree_t)} \pi_\mu(\mu_{t\ell}) \ d\mu_{t\ell}
  \propto \prod_{\ell \in \Leaves(\Tree_t)} \frac{b^a}{\Gamma(a)} \times \frac{\Gamma(a + A_{t\ell})}{(b + B_{t\ell})^{a + A_{t\ell}}},
\end{align*}
and the full conditional $\mu_{t\ell} \sim \log \Gam(a + A_{t\ell}, b + B_{t\ell})$. Following \citet{chipman2010bart}, these computations are sufficient to construct a Bayesian backfitting algorithm for the Gumbel regression model for $Z_i$.

To see how this is useful, suppose we want to model $[Y_i \mid X_i] \sim \Bernoulli(p_i)$ with the \cloglog\ link $g(p_i) = \log\{-\log(1 - p_i)\} = r(X_i)$, or equivalently $p_i = 1 - \exp(-e^{r(X_i)})$. Noting that $G(t) = 1 - \exp(-t e^{r(x)})$ is the cumulative distribution function (cdf) of an $\Exp(e^{r(x)})$ random variable, we see that this is equivalent to setting $Y_i = 1$ if $E_i \le 1$ and $Y_i = 0$ if $E_i > 1$ where $E_i \sim \Exp(e^{r(X_i)})$. We can then perform data augmentation to reduce inference for the \cloglog\ link to the exponential model by augmenting $[E_i \mid Y_i = 1, X_i] \sim \TExp(e^{r(X_i)}, 0, 1)$ and $[E_i \mid Y_i = 0, X_i] \sim \TExp(e^{r(X_i)}, 1, \infty)$, where $\TExp(\beta, a, b)$ denotes an $\Exp(\beta)$ random variable truncated to the interval $(a,b)$.

\begin{remark}
  Inference can also proceed with the log-log link $g(p) = \log\{-\log(p)\}$, although \citet{mccullagh1998generalized} note that this link is less popular due to it having a heavy right tail rather than a heavy left tail. To use the log-log link, we simply switch the coding of successes and failures. A similar trick can also be used with ordinal models, where the cumulative log-log link can be specified by reversing the ordering of the categories.
\end{remark}

\begin{remark}
  Strictly speaking, it is possible to perform Bayesian backfitting for the \cloglog\ binary model by only augmenting variables with $Y_i = 1$, as the likelihood contribution $\exp(-e^{r(X_i)})$ turns out not to obstruct the conjugacy of the log-gamma prior. Avoiding augmenting $E_i$ when $Y_i = 0$ can be desirable computationally if there are very few successes, as this limits the number of $E_i$'s we need to augment.
\end{remark}

\subsection{Proportional Hazards Ordinal Regression Model}
\label{sec:ordinal-data-models}

There are two commonly used frameworks for modeling ordinal outcome data: \emph{cumulative link} models and \emph{continuation ratio} (or \emph{sequential regression}) models. Interestingly, in the special case of the \cloglog\ link, these two frameworks coincide. This model is often referred to as a (grouped) \emph{proportional hazards} model because it satisfies the relationship
\begin{math}
\Pr(Y_i > k \mid X_i = x) = \{S_0(k)\}^{e^{r(x)}},
\end{math}
where $S_0(k)$ represents a baseline survival function. This formulation generalizes the structure of Cox's proportional hazards model in survival analysis to discrete data \citep{suresh2022survival}; when $r(x)$ has a BART prior, we refer to this model as the PHOBART (proportional hazards ordinal BART) model.





The cumulative link formulation of the proportional hazards model arises from a latent utility $Z_i$, which is categorized into an ordinal outcome based on a set of ordered \emph{cutpoints}, $-\infty = c_0 < c_1 < \cdots < c_{K-1} < c_K = \infty$. The observed ordinal outcome $Y_i$ corresponds to the number of thresholds $Z_i$ exceeds, formally given by $Y_i = \sum_k 1(Z_i > c_k)$. To derive the proportional hazards model, we define $Z_i = -r(X_i) + \epsilon_i$ with $\log \epsilon_i \sim G$ with $G(t) = 1 - \exp(-e^t)$ the Gumbel cdf. This leads to the likelihood and survival functions:
\begin{align*}
  \Pr(Y_i = k \mid X_i = x)
  &= G\{c_k + r(x)\} - G\{c_{k-1} + r(x)\}
  = \exp(-e^{c_{k - 1} + r(x)}) - \exp(-e^{c_k + r(x)}) \\
  \Pr(Y_i > k \mid X_i = x)
  &= \exp(-e^{c_k + r(x)}) = \big\{\exp(-e^{c_k})\big\}^{e^{r(x)}} = S_0(k)^{e^{r(x)}}.
\end{align*}
This verifies that the proportional hazards characterization of the model. Figure~\ref{fig:cumulative-link} gives a visual comparison of this model with the cumulative probit model.

\begin{figure}
  \centering
  \includegraphics[width=1\textwidth]{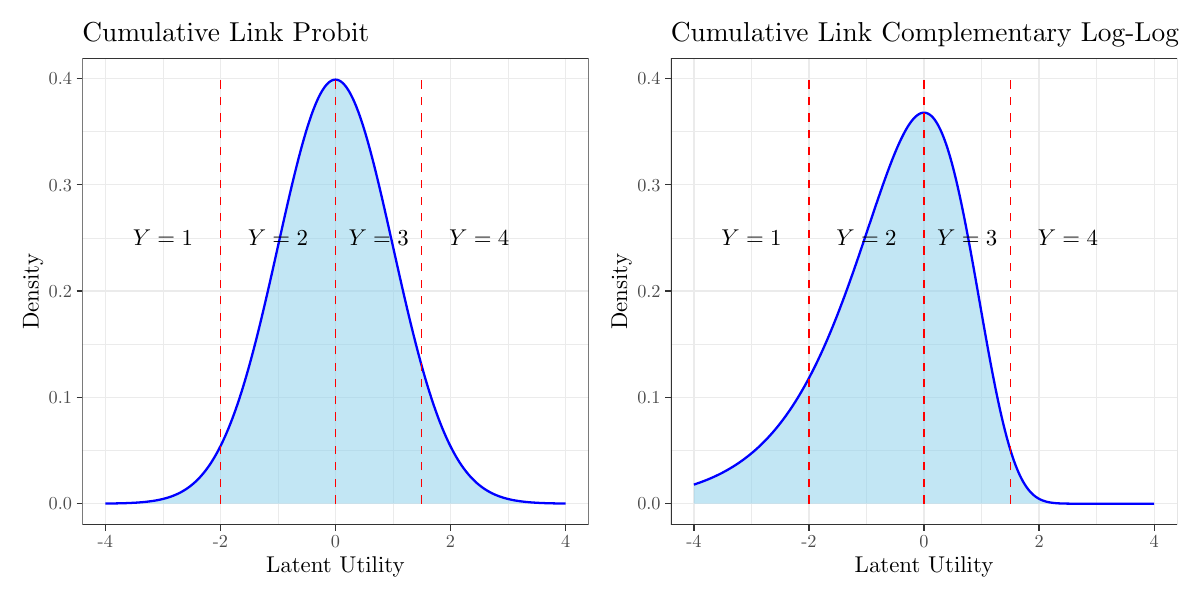}
  \caption{Visualization of the cumulative link model for ordinal data using the probit (left) and complementary log-log (right) link functions. The latent utility $Z_i$ is depicted along with the cutpoints $\gamma_k$ that determine the observed ordinal outcomes $Y_i$.}
  \label{fig:cumulative-link}
\end{figure}


The continuation ratio approach, in contrast, models the probability of ``stopping'' at a given category as $\Pr(Y_i = k \mid Y_i \ge k, X_i = x) = G\{\gamma_k + r(x)\}$ so that
\begin{align}
  \label{eq:continuation-ratio}
  \Pr(Y_i = k \mid X_i = x) &= G\{\gamma_k + r(x)\} \prod_{j < k} [1 - G\{\gamma_j + r(x)\}].
\end{align}
When $G^{-1}(\cdot)$ corresponds to the \cloglog\ link, \eqref{eq:continuation-ratio} can be rewritten as
\begin{align*}
  \exp\left\{-\left(\sum_{j < k} e^{\gamma_j}\right) e^{-r(x)}\right\} -
  \exp\left\{-\left(\sum_{j \le k} e^{\gamma_j}\right) e^{-r(x)}\right\}
    &= \exp\left(-e^{c_{k - 1} + r(x)}\right) - \exp(-e^{c_k + r(x)})
  \\&= G\{c_k + r(x)\} - G\{c_{k-1} + r(x)\},
\end{align*}
where $c_k = \log \sum_{j \le k} e^{\gamma_j}$. This establishes that the continuation ratio and cumulative link models coincide under the \cloglog\ link.

Unlike the probit and logit link functions, we can leverage the correspondence between the continuation ratio and cumulative link models to simplify inference on the cutpoints $c_k$. As shown in Web Appendix~\ref{sec:gibbs-sampler-ordinal}, the continuation ratio parametrization enables straightforward Gibbs sampling updates for the $\gamma_k$'s when a $\log \Gam(a_\gamma, b_\gamma)$ prior is used. The cutpoints can then be easily obtained by transforming the $\gamma_k$'s via $c_k = \log \sum_{j \le k} e^{\gamma_j}$. In contrast, efficient Gibbs samplers for the probit and logit link functions must explicitly account for the ordering constraint on the $c_k$'s, adding complexity to the inference process.

\subsection{Fully-Nonparametric Ordinal Model}

The cumulative link model, while convenient, imposes a parametric structure that may not always be appropriate. As an alternative, we propose a fully nonparametric continuation ratio model, defined as
\begin{math}
  \Pr(Y = k \mid Y_i \ge k, X_i = x) = 1 - \exp(-e^{\gamma_k + r(x, k)}),
\end{math}
where the nonparametric function $r(x, k)$ depends on both $x$ and $k$. This flexible structure allows the model to adapt to any data-generating mechanism for ordinal data. When $r(x, k)$ has a BART prior, we refer to this model as the NPHOBART (non-proportional hazards ordinal BART) model.

This approach also conveniently allows us to shrink the model toward the cumulative link structure by using the automatic relevance determination prior of \citet{linero2018bayesian}. Let $s = (s_1, \ldots, s_P, s_{P + 1})$ represent the prior probabilities that a given splitting rule in the ensemble uses predictor $j$ to construct a split, where $j = P + 1$ corresponds to a split on the category $k$. \citet{linero2018bayesian} proposed a $\Dirichlet(\alpha, \ldots, \alpha)$ prior for $s$, with $\alpha \propto P^{-1}$, to facilitate variable selection in high-dimensional settings. To adapt this to our setting, we instead use $s \sim \Dirichlet(1, \ldots, 1, w)$, where $w$ is a small positive constant. This formulation allows the prior to favor not splitting on the category $k$, effectively reverting the model to the proportional hazards structure. This hierarchical framework provides a mechanism for balancing flexibility and parsimony, enabling the model to adapt to the data while incorporating a natural shrinkage toward the simpler proportional hazards model when appropriate.

\subsection{Data Augmentation for Ordinal Models}
\label{sec:data-augmentation}

Data augmentation for the cumulative link model can proceed similarly to binary regression models. Starting from \eqref{eq:continuation-ratio}, we note that if \( Z \sim \TExp(e^{\gamma_k + r(x)}, 0, 1) \), then \( Z \) has the density
\[
  f(z \mid \gamma, r) = \frac{\exp\{\gamma_k + r(x) - z e^{\gamma_k + r(x)}\}}{G\{\gamma_k + r(x)\}}.
\]
Truncated exponential random variables can be sampled efficiently using the probability integral transform \citep[][Section 2.1]{casella2002statistical}. For each \( Y_i < K \), augmenting \( Z_i \sim \TExp(e^{\gamma_{Y_i} + r(X_i)}, 0, 1) \) leads to the augmented likelihood for observation \( i \):
\[
  \exp\left( 1(Y_i < K) \{\gamma_{Y_i} + r(X_i) - Z_i e^{\gamma_{Y_i} + r(X_i)}\} - \sum_{k = 1}^{Y_i - 1} e^{\gamma_k + r(X_i)} \right),
\]
which can be rewritten as
\[
  \exp\left( \{\gamma_{Y_i} + r(X_i)\} 1(Y_i < K) - e^{r(X_i)} \left[ \sum_{k = 1}^{Y_i - 1} e^{\gamma_k} + Z_i e^{\gamma_{Y_i}} 1(Y_i < K) \right] \right).
\]
This form matches the Poisson-gamma likelihood described by \citet{hill2020bayesian}, making it amenable to efficient inference using a Bayesian backfitting algorithm.

A similar strategy applies to the fully nonparametric model, which replaces \( r(x) \) with \( r(x, k) \). In this case, we augment \( Z_i \sim \TExp(e^{\gamma_{Y_i} + r(X_i, Y_i)}, 0, 1) \) for \( Y_i < K \). Further implementation details are provided in Web Appendix~\ref{sec:gibbs-sampler-ordinal}.

\subsection{Stick Breaking Processes}
\label{sec:stick}

In addition to its utility for modeling ordinal data, the \cloglog\ link can be leveraged to construct a \emph{stick-breaking process} for Bayesian nonparametric modeling of conditional random measures. Specifically, we consider conditional distributions of the form
\begin{align}
  \label{eq:stick-break}
  F_x(dy) = \sum_{k = 1}^\infty \varpi_k(x) \, \delta_{\zeta_k(x)} \qquad \text{where} \qquad
  \varpi_k(x) = \{1 - \exp(-e^{\gamma_k + r(x,k)})\} \prod_{j < k} \exp(-e^{\gamma_j + r(x,j)}),
\end{align}
with a modified BART prior placed on \( r(x,k) \). Here, the mixture components \( \zeta_k(x) \) are sampled independently from a base distribution \( H(d\zeta) \). This framework parallels the structure of probit \citep{rodriguez2011nonparametric} and logit \citep{rigon2021tractable} stick-breaking processes but with the \cloglog\ link.

We consider two variants of \eqref{eq:stick-break}. The simpler variant assumes \( r(x,k) \equiv r(x) \), meaning the nonparametric function \( r \) does not include an interaction between \( x \) and \( k \). We refer to this simplified random distribution as the \emph{proportional hazards stick-breaking process} (PHSBP). Notably, the PHSBP can be cast as a type of dependent Dirichlet process \citep[DDP,][]{maceachern2000dependent} with covariate dependent weights.

\begin{theorem}
  \label{prop:dp-connection}
  Suppose that $\gamma_k \sim \log \Gam(1, 1)$ and let $F_x$ be defined as in \eqref{eq:stick-break}. Then conditional on $x$ and $r(x)$ we have
  \begin{math}
    F_x \sim \DP(e^{-r(x)}, H)
  \end{math}
  where $\DP(\alpha, H)$ denotes a Dirichlet process with concentration parameter $\alpha$ and base distribution $H$.
\end{theorem}

\begin{proof}[Proof sketch]
  For any exponential random variable $E_i \sim \Exp(\lambda)$ with rate $\lambda$, we have $e^{-Z_i} \sim \Beta(\lambda, 1)$. Noting that $e^{\gamma_k + r(x)} \sim \Exp(e^{-r(x)})$, our setup then matches the Sethuramann stick-breaking construction of the Dirichlet process \citep{sethuraman1994constructive}.
\end{proof}


Unfortunately, while the weights of the PHSBP are covariate-dependent, Theorem~\ref{prop:dp-connection} implies that this dependence is limited to allowing the concentration parameter of the resulting Dirichlet process to vary with \(x\). This type of covariate dependence is of limited practical interest. To address this limitation, we propose the \emph{non-proportional hazards stick-breaking process} (NPHSBP), which introduces an interaction between \(x\) and \(k\). To incorporate the mixture component index \(k\) in a meaningful way, we use splitting rules of the form \( [p(k - 1) \le C_b] \), where \(p(k)\) is the cumulative distribution function (CDF) of any discrete random variable on \( 0, 1, 2, \ldots \), and \(C_b \sim \Uniform(0,1)\). This prioritizes re-weighting the first few mixture components, placing less emphasis on components associated with larger \(k\). In our illustrations, we set \( p(k) \) to be the CDF of a geometric random variable with a success probability of \( 1/3 \).

The NPHSBP can also be related to the truncated stick-breaking processes defined by \citet{ishwaran2001gibbs}, as shown in the following theorem.

\begin{theorem}
  \label{prop:sbp-connection}
  Let \( \mathcal{P}_K(\bm{a}, \bm{b}) \) denote the stick-breaking process defined by \citet{ishwaran2001gibbs} with \( K \) components, where \( \bm{a} \) and \( \bm{b} \) are sequences of positive numbers. Suppose \( \gamma_k \sim \log \Gam(1, 1) \), and let \( F_x \) be defined as in \eqref{eq:stick-break}. Then, conditional on \( x \) and \( r(x,k) \), we have
  \begin{math}
    F_x \sim \mathcal{P}_\infty\{\bm{a}(x), \bm{b}(x)\},
  \end{math}
  where \( a_k(x) = 1 \) and \( b_k(x) = e^{-r(x,k)} \). If the NPHSBP is instead truncated at a finite \( K \), then \( F_x \sim \mathcal{P}_K\{\bm{a}(x), \bm{b}(x)\} \).
\end{theorem}

In our illustrations, we use the NPHSBP to perform density regression with the mixture model
\begin{align*}
  f(y \mid x) = \int \Normal\{y \mid \mu + h(x), \sigma^2\} \, F_x(d\mu, d\sigma)
  = \sum_{k = 1}^\infty \varpi_k(x) \, \Normal\{y \mid \mu_k + h(x), \sigma_k^2\}
\end{align*}
where $h(x)$ is also given a BART prior. By explicitly modeling overall location of $y$ with $h(x)$, this approach encourages the mixture model components $(\varpi_k, \mu_k, \sigma^2_k)$ to focus on capturing changes in the residual distribution.

\paragraph{Inference Via Gibbs Sampling}
Inference for the NPHSBP alternates between (i) updating the mixture components $C_i$ for each of the observations, (ii) updating the parameters $(\gamma_k, r)$ for the weights conditional the $C_i$'s, and (iii) updating the regression parameters $(\mu_k, \sigma_k, h)$ for the outcome. A detailed algorithm Gibbs sampler is given in Web Appendix~\ref{sec:gibbs-density-regression}. The main simplification provided by our use of the \cloglog\ link is that the second step involves augmenting only one latent variable for each observation, as opposed to needing to sample $C_i$ latent variables for the probit/logit stick-breaking processes.

\paragraph{Prior Specification}
Our default prior for the NPHSBP uses default BART priors for $h(x)$ and $r(x)$, with leaf node standard deviations of $1 / \sqrt T$ for both models. Prior to fitting the model, the outcomes $Y_i$ are standardized to have mean $0$ and variance $1$. The base distribution of the measure takes $\mu_k \sim \Normal(0, \sigma_0^2)$ and $\sigma^{-2}_k \sim \Gam(a_\sigma, b_\sigma)$. The hyperparameters are given the priors $a_\sigma \sim \Gam(4, 2)$, $b_\sigma \sim \Gam(4, 2)$, $\sigma_0^{-2} \sim \Gam(1, 1)$.

\subsection{Proportional and Non-Proportional Hazards Survival Models}

The methods we have developed for cumulative link ordinal regression models naturally extend to survival analysis, leading to both proportional hazards (PH) and non-proportional hazards (NPH) models for time-to-event data. When the categories $k = 1,\ldots,K$ are interpreted as indexing a discrete set of possible event times $t_1, \ldots, t_K$ our PH model becomes equivalent to a discrete-time proportional hazards model \citep{kalbfleisch2002statistical, suresh2022survival} with
\begin{math}
  \Pr(T_i \ge t_k \mid X_i = x) = \{S_0(t_k)\}^{e^{r(x)}}.
\end{math}
In the context of survival analysis, our NPH model offers a computationally efficient alternative to the method proposed by \citet{sparapani2016nonparametric}. While both approaches use stick-breaking constructions, Sparapani's method requires augmenting a high-dimensional latent variable for each observation, making it more computationally intensive.

In this section, we present a continuous-time analog of the discrete-time proportional hazards model. Unlike discrete-time survival analysis, our continuous-time survival models avoid the need to augment any latent exponential random variables. A variant of this model adapted to \emph{relative} survival is developed by \citet{basak2024relative}.

We assume the usual survival analysis setup, with $T_i$ as the time to death and $C_i$ as a censoring time. Rather than observing $(T_i, C_i)$, we observe $Y_i = \min(T_i, C_i)$ and $\delta_i = 1(T_i < C_i)$. Under this setup, the proportional hazards model assumes a form for the survival function $S(t \mid x) = \Pr(T_i > t \mid X_i = x)$ of
\begin{math}
  S(t \mid x) = \exp\left\{ -e^{r(x)} \, \int_0^t \lambda_0(u) \, du  \right\}
\end{math}
The function $\lambda_0(t)$ is referred to as the baseline hazard function. This model is referred to as a proportional hazards model because the hazard function $\lambda(t \mid x) = \frac{d}{dt} -\log S(t \mid x)$ factors as $\lambda_0(t) \, e^{r(x)}$. To obtain a flexible model for the baseline hazard function, we specify a piecewise exponential model:
\begin{align*}
  \lambda_0(t) = \sum_{b = 1}^B 1(t_{b-1} \le t < t_b) \, \lambda_b
\end{align*}
where $0 = t_0 < t_1 < \cdots < t_{B-1} < t_B = \infty$. Under this model, assuming that the censoring time $C_i$ is independent of $(T_i, X_i)$, the likelihood for observation $i$ reduces to:
\begin{align}
  \label{eq:likelihood}
  \exp\left\{ \delta_i \, \log \lambda_{B_i}  + \delta_i \, r(X_i) - e^{r(X_i)} \sum_{b = 1}^{B_i - 1} \lambda_b (t_b - t_{b-1}) - e^{r(X_i)} \lambda_{B_i} (Y_i - t_{B_i - 1}) \right\}
\end{align}
where $B_i$ is the indicator for the ``bin'' that $Y_i$ falls in, with $B_i = b$ if $t_{b-1} \le Y_i < t_b$.

The non-proportional hazards variant of this model replaces $r(X_i)$ with $r(X_i, B_i)$, so that the hazard becomes:
\begin{math}
  \lambda(t \mid x) = \sum_{b = 1}^B 1(t_{b-1} \le t < t_b) \, \lambda_b \, \exp\{r(x,b)\}
\end{math}
This modification allows $\lambda(t \mid x)$ to vary with the covariates in a fully nonparametric fashion. For this model, the likelihood is:
\begin{align}
  \label{eq:likelihood-nph}
  \exp\left\{ \delta_i \, \log \lambda_{B_i} + \delta_i \, r(X_i, B_i) - \sum_{b = 1}^{B_i - 1} \lambda_b \, e^{r(X_i, b)} (t_b - t_{b-1}) -  \lambda_{B_i} \, e^{r(X_i, B_i)} (Y_i - t_{B_i - 1})\right\}.
\end{align}

\paragraph{Inference Via Gibbs Sampling}

Both the PH and NPH survival models admit simple Bayesian backfitting algorithms that alternate between: (i) updating the parameters $(\lambda_1, \ldots, \lambda_B)$ and (ii) updating $r(x)$ or $r(x,b)$. To implement this, we use the fact that equations \eqref{eq:likelihood} and \eqref{eq:likelihood-nph} lead to the same Poisson-gamma form described by \citet{hill2020bayesian}. Full details of the algorithm are provided in Web Appendix~\ref{sec:gibbs-survival}. As demonstrated in the Web Appendix, a key difference between the PH and NPH Gibbs samplers is their computational complexity: the PH algorithm allows the likelihood contribution of observation $i$ to tree $t$ to be computed in constant time, while the NPH algorithm requires $O(B_i)$ computations. This provides additional motivation for keeping $B$ modest in the NPH model.

\paragraph{Prior Specifications}
For the $\lambda_b$'s, we specify the prior $\lambda_b \stackrel{\text{iid}}{\sim} \Gam(a_\lambda, b_\lambda)$. In our illustrations, we set $a_\lambda = b_\lambda = 1$ and take $B = N^{1/3}$ (matching the order of the Freedman-Diaconis rule for selecting the number of bins in a histogram). The $t_b$ values are chosen so that an equal proportion of the uncensored event times lies in each bin. For $r(x)$ and $r(x,b)$, we use the same default prior specification as used in the ordinal regression models.

\section{Rates of Posterior Contraction}
\label{sec:posterior-concentration}

One way to compare Bayesian nonparametric priors is by studying their posterior contraction rates. This framework quantifies how fast a posterior distribution $\Pi(dr \mid \text{Data}_n)$ concentrates around an underlying ``true'' nonparametric regression function $r_0(x)$, which can be used to guide prior selection, particularly in high or infinite-dimensional settings. In recent years, there has been a surge of interest in studying posterior contraction rates for BART models \citep{linero2018bayesian, rovckova2020posterior, orlandi2021density, linero2022bayesian, jeong2023art, li2023adaptive, saha2023theory, linero2024generalized}. In this section, we establish posterior contraction rates for the models proposed in Section~\ref{sec:PH}, in particular for the cumulative link \cloglog\ model.

We start with the simplifying assumption that the predictors $X_i$'s are supported on $\scriptX = [0,1]^P$. We impose the following assumption on the unknown parameter $r_0: \scriptX \to \R$. 



\paragraph{Assumption F} The function $r_0$ is $\alpha$-H\"{o}lder smooth for some $\alpha \in (0,1]$, and depends on only $D_0$ coordinates.

\begin{remark}
  Formally, $r_0$ is in a H\"{o}lder ball $H(\alpha, C)$ if $r_0(x)$ and all of its mixed derivatives of order less than $\lfloor \alpha \rfloor$ are bounded by $C$, and all $\lfloor \alpha \rfloor$ ordered derivatives of $r_0$ are Lipschitz with exponent $\alpha - \lfloor \alpha \rfloor$, where $\lfloor \alpha \rfloor$ is the largest integer strictly less than $\alpha$. It is possible to relax Assumption F to allow $\alpha > 1$ (or to instead consider low-order interaction structure) using soft decision trees \citep{linero2018abayesian}. Our proof strategy can also accommodate $P$ diverging nearly-exponentially in $n$, provided that $D_0$ remains bounded, at the expense of an additional variable selection term $\sqrt{ \frac{D_0 \log P}{n}}$ in the error rates, provided that an appropriate sparsity-inducing prior is used. The following theorem summarizes our results.
\end{remark}

\begin{theorem}
  \label{thm:posterior-concentration}
  Suppose $(Y_i, X_i : i = 1,\ldots, n)$ are iid pairs with $X_i \sim F_X$ and $Y_i \in \{1,\ldots,K\}$ with $\Pr(Y_i > k \mid X_i, \gamma_0, r_0) = \exp(-e^{\gamma_{0k} + r_0(x)})$ for $k = 1, \ldots, K-1$. Assume that $r_0$ satisfies Assumption F, the prior $\Pi(dr, d\gamma)$ satisfies Condition P in Web Appendix~\ref{sec:proof-of-theorem}, and the $\gamma_j$'s have a $\log \Gam(a_\gamma, b_\gamma)$ prior. Then there exists a constant $M > 0$ such that
  \begin{align*}
    \Pi\{d(\theta, \theta_0) \geq M\epsilon_n \mid (Y_i, X_i: i = 1, \dots, n)\} \to 0 \ \text{ in probability},
  \end{align*}
  where $\theta = (r, \gamma)$ and $d(\theta, \theta_0) = \int \sum_{y = 1}^K |p_\theta(y \mid x) - p_{\theta_0}(y \mid x)| \ F_X(dx)$ is the $F_X$-integrated total variation distance, $\Pi\{d\theta \mid (Y_i, X_i: i = 1, \dots, n)\}$ is the posterior distribution of $\theta$, and $\epsilon_n = (\log n / n)^{\alpha / (2\alpha + D_0)}$.

\end{theorem}

We provide a detailed proof of Theorem \ref{thm:posterior-concentration} in the supplementary material. This theorem implies that our BART-based posterior distribution concentrates around the true data generating mechanism at an optimal rate up-to logarithmic terms.

\begin{remark}
  We mainly provide this result for illustrative purposes, and similar results can likely also be established for the other models, adapted to their contexts. The main technical tools we use for establishing these results are prior thickness conditions and model entropy conditions established by \citet{jeong2023art}, which allow us to verify the sufficient conditions of \citet{ghosal2000convergence}.
\end{remark}

\section{Illustrations}
\label{sec:illustrations}

\subsection{Density Regression}
\paragraph*{Simulation Experiment.}
We evaluate the performance of NPHSBP using a simulation framework adapted from \citet{dunson2007bayesian}. The response variable $Y_i$ is generated from conditional a mixture distribution
$$
Y_i \sim e^{-2x} \Normal(x, 0.1^2) + (1 - e^{-2x}) \Normal(x^4, 0.2^2),
$$
given $X_i = x$. We set $n = 500$ and $X_i \sim \Uniform(0,1)$. This setup is challenging because of the rapidly changing shape of the conditional density across the range of $x$, making accurate estimation difficult, especially with limited data in specific local regions.

Figure~\ref{fig:dunson} shows the results of fitting the NPHSBP mixture model to a sample from this data generating mechanism. The first five panels display the true conditional response density $f(y \mid x)$ (dashed black line), the posterior mean estimate (solid blue line), and 95\% pointwise credible bands (shaded blue area) for five selected percentiles (10th, 25th, 50th, 75th, and 90th) of the empirical distribution of covariate $x$. The last panel shows the observed data (green points), true mean regression function $\E(Y \mid X = x)$ (dashed black line), posterior mean estimate (solid blue line), and 95\% pointwise credible band (shaded blue area) across the entire range of $x$.

The results show that NPHSBP effectively captures the complex, non-linear relationship between $x$ and $y$. The posterior mean estimates closely approximate the true conditional densities and mean regression function across different regions of the covariate space. The true functions are contained within the 95\% credible bands. Notably, the NPHSBP mixture model accurately estimates the changing shape of the conditional density as $x$ increases, transitioning from a unimodal to a bimodal distribution.

\begin{figure}
  \centering
  \includegraphics[width=1\textwidth]{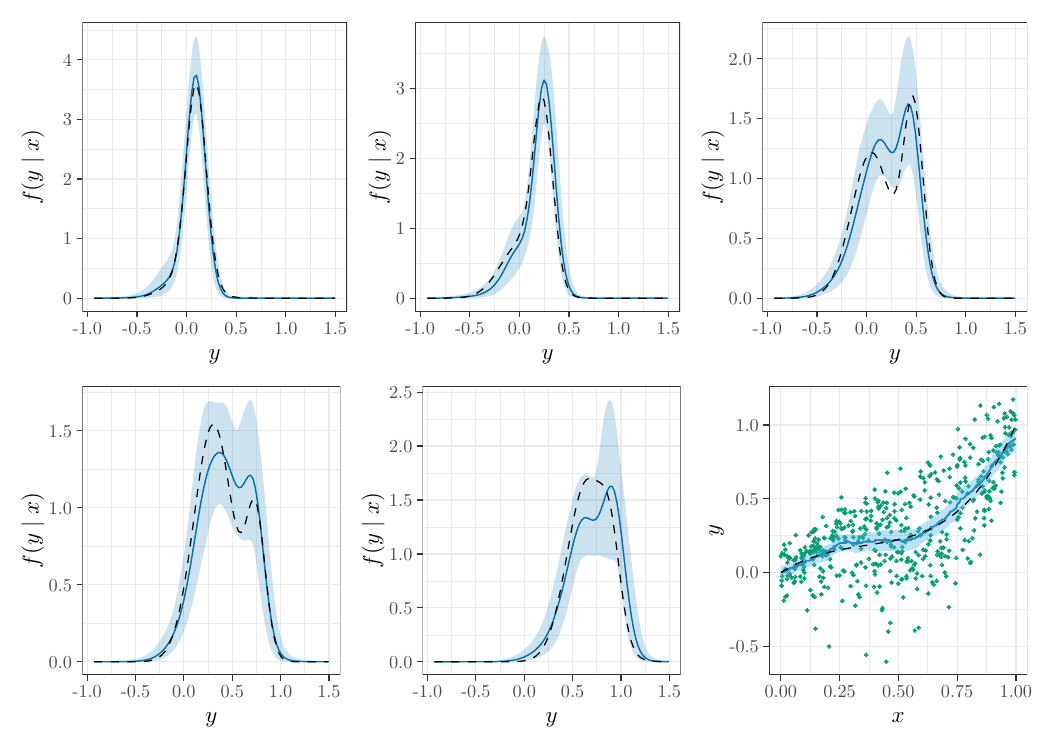}
  \caption{The first five panels display the true conditional response density $f(y \mid x)$ [dashed black line], posterior mean estimate [solid blue line], and 95\% pointwise credible bands [shaded blue area] for five selected percentiles of covariate $x$ at 10\%, 25\%, 50\%, 75\%, and 90\%. The last panel displays the observed data [green points] and true mean regression function $\E(y \mid x)$  [dashed black line], posterior mean estimate [solid blue line], and 95\% pointwise credible band [shaded blue area].}
  \label{fig:dunson}
\end{figure}

\paragraph*{Application to MEPS Data.}
We analyze data from the 2015 Medical Expenditure Panel Survey (MEPS), an ongoing U.S. survey that collects data on
individuals, families, their medical providers, and employers, with a particular focus on healthcare costs and uses.
Existing literature extensively studied the relationship between socioeconomic status, education, and obesity.
The effect of educational attainment on obesity is complex and known to be moderated by regional income, gender,
and other factors \citep{cohen2013educational}. We analyze this relationship using a subset of MEPS data, focusing
on $9,426$ women, controlling for log-income (as a percentage of the poverty line), age, and race. Based on existing
research, we anticipate an inverse relationship between educational attainment and obesity levels in this demographic.

We fit the proposed NPHSBP model and a total of $4,000$ posterior samples are generated using the MCMC algorithm detailed
in Web Appendix~\ref{sec:gibbs-density-regression}. We discard the first $2,000$ iterations as initial burn-in and
use the remaining $R = 2,000$ Monte Carlo samples for the results. Figure \ref{fig:density_meps} (top panel)
illustrates the estimated density of BMI as educational attainment varies from less-than-high-school to graduate degree
for white women, keeping other covariates fixed at their median values. We see the modal BMI value remains relatively
stable across education levels. However, as educational attainment increases, the BMI distribution becomes more peaked and
more highly right-skewed, suggesting that highly educated women are less likely to be highly obese. Specifically,
the relationship indicates that higher educational attainment is associated with a reduced probability of severe obesity.
As an example, we display the difference in BMI density between bachelor's degree and high school diploma holders in
Figure \ref{fig:density_meps} (bottom panel), which clearly shows the reduced extreme obesity risk for bachelor's degree
holders compared to high school graduates.

\begin{figure}
  \centering
  \includegraphics[width=1\textwidth]{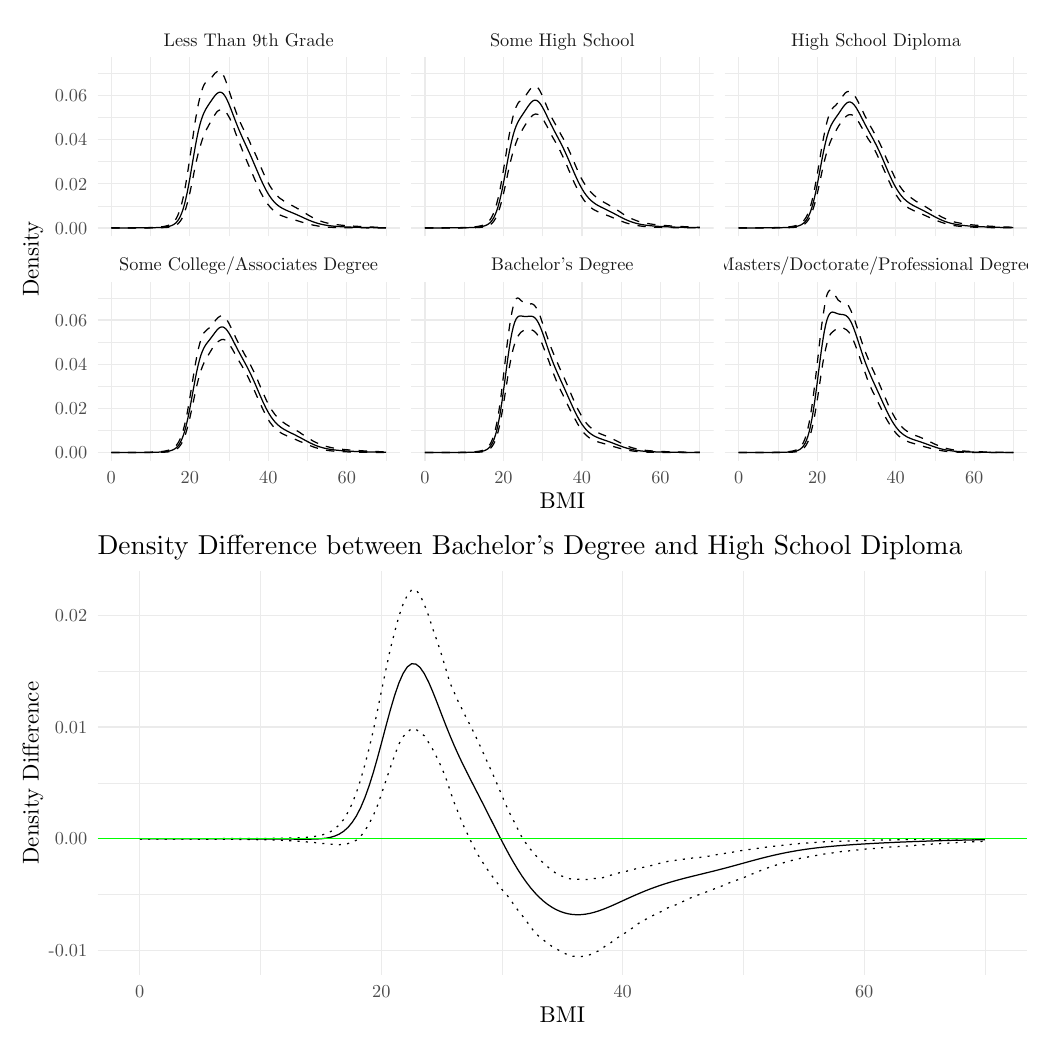}
  \caption{(Top panel) Estimated BMI density and 95\% pointwise credible bands for different educational levels for white women, fixing other covariates at their median values, based on MEPS data. (Bottom panel) Difference in BMI density between bachelor's degree and high school diploma holders.}
    \label{fig:density_meps}
\end{figure}


\subsection{Ordinal Regression}
We demonstrate the application of the proposed models, PHOBART and NPHOBART, for ordinal regression using data from the
$2021$ Medical Expenditure Panel Survey (MEPS). This analysis focuses on the relationship between depression levels,
treated as the ordinal response variable, and key demographic and socioeconomic variables, including age, income,
education level, marital status, and sex as covariates. Depression levels are categorized on a four-point scale: $0$
(not at all), $1$ (several days), $2$ (more than half the days), and $3$ (nearly every day). For convenience, we refer to these categories as \textit{no}, \textit{mild}, \textit{moderate}, and \textit{severe} depression levels. The data set consists of $14,676$ observations.

We start by fitting the PHOBART and NPHOBART models using the MCMC
algorithm outlined in Web Appendix~\ref{sec:gibbs-sampler-ordinal}, with $T = 50$ trees. To benchmark our proposed
models, we compare their performance against Bayesian linear Cumulative Link models: one with a probit link
(ProbitLM) and the other with a complementary log-log link (CLogLogLM), both implemented using the \texttt{brms}
package \citep{brms_package} in \texttt{R}. To ensure convergence, we discard the first $2,500$ iterations as burn-in and use the
remaining $2,500$ posterior samples for inference. For model comparison, we compute the leave-one-out expected log predictive
density (ELPD), a comprehensive measure of goodness-of-fit. Details on the calculation and
interpretation of this statistic can be found in \citet{vehtari2017practical}. The ELPD values
for the PHOBART and NPHOBART models are similar, with values of $-10551.8$ and $-10548.8$, respectively,
suggesting both models offer a comparable fit to the data. In contrast, the ELPD values for the ProbitLM and
CLogLogLM models are $-10623.9$ and $-10608.9$, respectively, indicating that the BART models outperform these alternatives. Notably, CLogLogLM performs better than ProbitLM, suggesting that the complementary log-log link is more appropriate for this data. A similar pattern is observed when comparing models using $5$-fold cross-validation, where the held-out deviance is averaged across $10$ different splits. The results, given in Table \ref{tab:model-comparison-ordinal}, show that the PHOBART model outperforms both CLogLogLM and ProbitLM, with a higher performance gap for ProbitLM. This further supports the suitability of the complementary log-log link over the probit link for this data.

\begin{table}[ht]
  \centering
  \begin{tabular}{lrr}
    \toprule
    Model      & Deviance Difference      \\
    \midrule
   CLogLogLM  & -39.82 \\
   ProbitLM    & -78.38     \\
    \bottomrule
  \end{tabular}
  \caption{Comparison of model performance. The deviance difference is computed as the deviance of the PHOBART model minus the deviance of the given model. Negative values indicate that the PHOBART model performs better than the alternatives.}
 \label{tab:model-comparison-ordinal}
\end{table}

\begin{figure}[t]
  \centering
  \includegraphics[width=1\textwidth, height = 0.5\textheight]{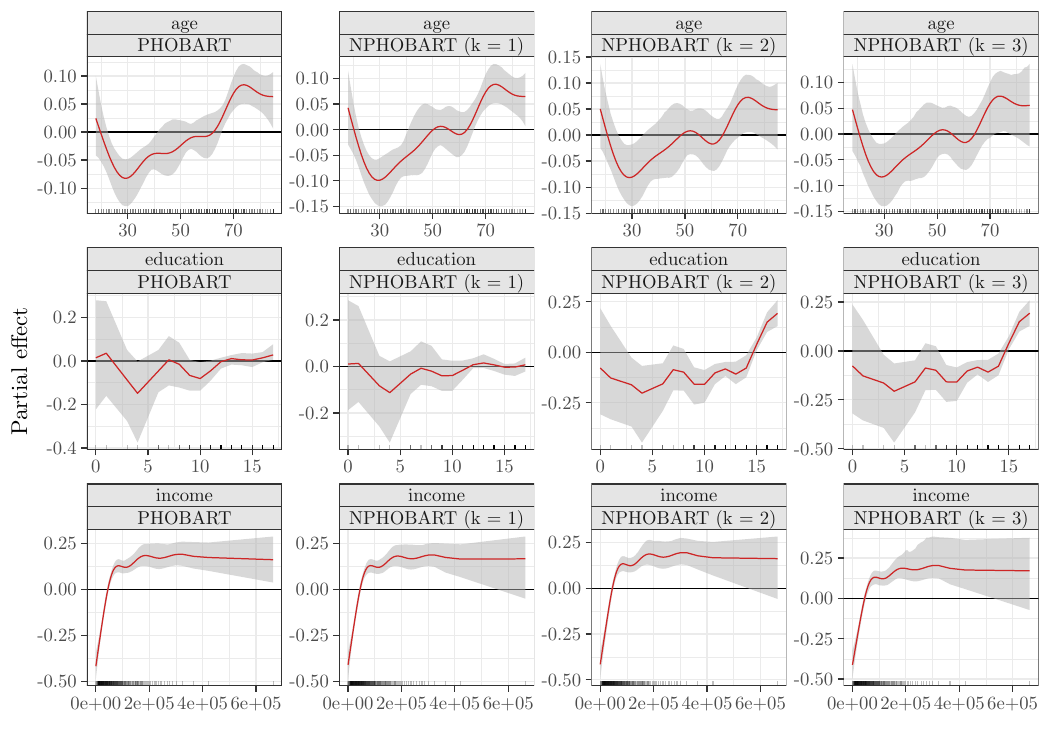}
  \caption{Posterior partial‐effect functions for \textit{age} (top row),
    \textit{education} (middle row), and \textit{income} (bottom row) under the
    PHOBART model (left column) and under the NPHOBART model for each ordinal
    cutoff $k=1,2,3$ (three right columns). Each panel shows the
    projection‐based mean partial effect (solid red line) along with
    95\% credible bands (shaded gray).}
    \label{fig:ordinal_meps_pe}
\end{figure}

Next, We use the posterior projection strategy of \citet{woody2021model} to summarize the posterior distribution of $r(x)$ for PHOBART and compute its projection onto an interpretable family $\sQ$ as $\widetilde r(x) = \arg \min_{q \in \sQ}\|r-q\|$. We consider the family of additive models $q(x) = \sum_j q_j(x_j)$ as the choice of $\sQ$ in our illustrations. Similarly, we compute projections of $r(x,k), k = 1,\dots,K-1$ for NPHOBART, wherein the number of depression levels $(K)$ is $4$. Figure \ref{fig:ordinal_meps_pe} illustrates the partial effects of age, education, and income on depression levels. This figure shows an increasing effect of income on $r(x)$, with the effect stabilizing after reaching a certain income threshold, which suggests that while higher income is generally linked to lower depression, the benefits may level off after a certain point. The effect of education changes from $k=1$ to $k=2,3$. Higher education ($3$ or more years of college) is associated with lower depression when moving from \textit{mild} to \textit{moderate} or from \textit{moderate} to \textit{severe} depression levels. The relationship between age and mental health is complex and non-linear. The partial effect on $r(x)$ shows a sharp decline from ages $20$ to $30$, followed by a steady increase until approximately age $70$, and a slight decline thereafter. This pattern suggests a pronounced increase in depression during early adulthood (ages $20–30$), a gradual rise in happiness through later adulthood (ages $30–70$), and a mild downturn in mental health during the final stages of life.

\begin{figure}[ht]
  \centering
  \includegraphics[width=1\textwidth, height=0.48\textheight]{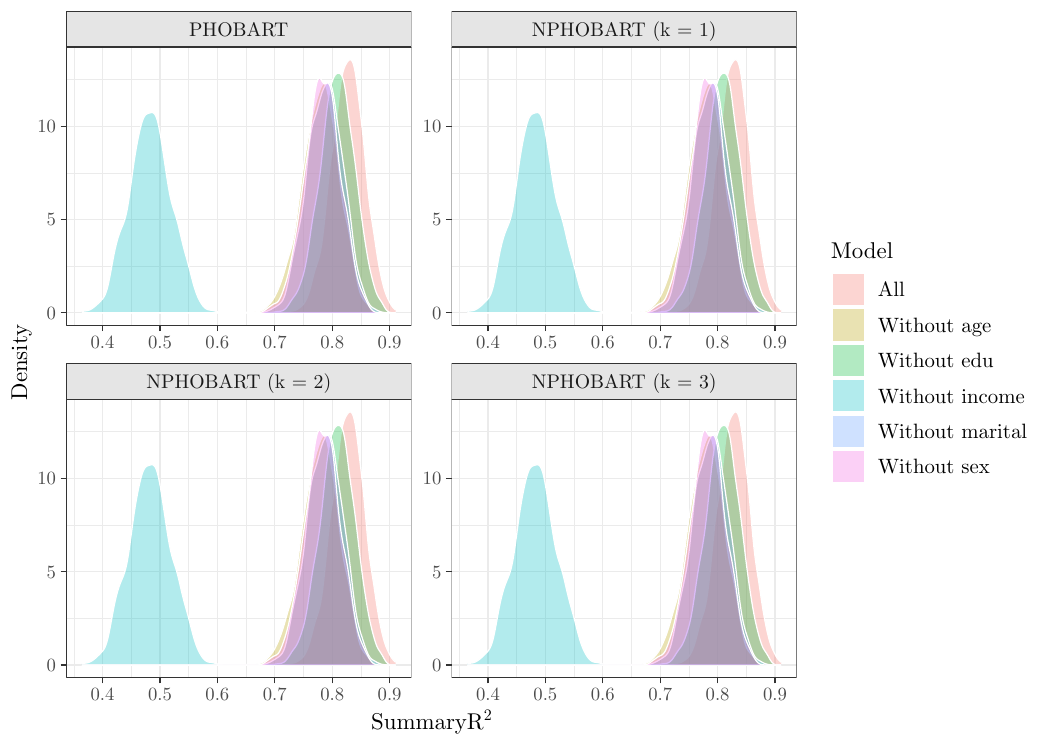}
  \caption{Posterior distributions of the summary $R^2$ statistics for the projection‐based summary model under PHOBART (top‐left) and NPHOBART ($k = 1, 2, 3$, remaining panels). The “All” distribution corresponds to using all predictors, while each other curve shows a reduced model obtained by omitting a single predictor (age, education, income, marital status, or sex).}
    \label{fig:ordinal_meps_rsq_all}
\end{figure}

We also generate the projections by leaving out one predictor at a time from the summary model $Q$ to assess its importance in forming predictions. We compute the summary $R^2$ for each of these summary models, where summary $R^2$ is defined as
$$
\text{summary } R^2 = \frac{\sum_i \{r(X_i) - \widetilde r(X_i)\}^2}{\sum_i \{r(X_i) - \bar r\}^2},
$$
where $\bar r = \frac{\sum_i r(X_i)}{N}$.
This $R^2$ measures the proportion of variability in $r(x)$ explained by the projection onto $Q$. Figure~\ref{fig:ordinal_meps_rsq_all} shows the posterior distributions of the projection‐based summary \(R^2\) statistic for both PHOBART and NPHOBART, comparing the full model (``All'') against models obtained by omitting each predictor (age, education, income, marital status, or sex) one at a time. We first see that the ``All'' $R^2$ is quite large, indicating that the effects of interaction effects is small relative to the main effects. Next, a substantial leftward shift in the \(R^2\) distribution indicates a larger drop in predictive accuracy for the corresponding reduced model. Notably, omitting income from the summary model causes the largest decrease in \(R^2\), which illustrates income’s outsized contribution to explaining depression levels in this data set. By contrast, excluding other predictors has a more modest effect on the \(R^2\) distributions, suggesting that although age, education, marital status, and sex each play a role, they are not as critical as income in capturing the overall variation in depression outcomes.

\subsection{Survival Regression}

We conclude with an illustration of the PH and NPH survival regression models on the \texttt{LeukSurv} dataset
available in the R package \texttt{spBayesSurv}. The outcome in this dataset is the survival time of leukemia
patients, and our goal is to understand how this survival time relates to the Townsend index (a measure of
societal deprivation), with other covariates of interest being age, sex, and overall white blood cell count.

We begin by comparing the PH approach to two competing methods:
\begin{enumerate}
\item[(i)] Cox Linear: a linear Cox model, $\lambda(t \mid x) = \lambda_0(t) \, e^{x^\top \beta}$ with $\lambda_0(t)$ modeled as piecewise continuous.
\item[(ii)] Weibull: the same as the Cox model, except $\lambda_0(t)$ is modeled parametrically as a Weibull hazard and with the continuous variables included as natural cubic splines.
\end{enumerate}

We compare these models by $5$-fold cross-validation using the held-out deviance, i.e., $\sum_i - 2 \log p(y_i \mid x_i, \delta_i, \Data_{-\text{fold}(i)})$ where $p(y \mid x, \delta, \Data)$ denotes the predictive density and $\Data_{-\text{fold}(i)}$ denotes the dataset that excludes the fold that observation $i$ is a member of. We then averaged this across 10 different splits into folds. Results are given in Table~\ref{tab:model-comparison}, where we see that the proportional hazards model outperforms both models.

\begin{table}[ht]
  \centering
  \begin{tabular}{lrr}
    \toprule
    Model      & Deviance Difference      \\
    \midrule
    Cox Linear  & -1.91 \\
    Weibull    & -217.00         \\
    \bottomrule
  \end{tabular}
  \caption{Model performance comparison under the proportional hazards assumption on the leukemia data \label{tab:model-comparison}. The deviance difference is computed as the deviance of the PH BART model minus the deviance of the given model.
}
\end{table}

Next, we compare the PH model to the NPH survival model. We again use held-out deviance to compare the models.
Results are given in Table~\ref{tab:looic-comparison} and find that the PH model performs notably worse than the
NPH model.

\begin{table}[ht]
  \centering
  \begin{tabular}{lr}
    \toprule
    Model & Deviance \\
    \midrule
    Non-Proportional Hazards & 11832.3 \\
    Proportional Hazards & 11895.4 \\
    \bottomrule
  \end{tabular}
  \caption{Comparison of the NPH and PH BART models on the leukemia data.\label{tab:looic-comparison}}
\end{table}

For brevity, we only interpret the model fit for the PH model; a more detailed analysis for the NPH model is given in the Supplementary Material, as well the posterior distribution of the baseline hazard for the PH model. In Figure~\ref{fig:addsum} we display the posterior distribution of the partial effect of each of the covariates on $r(x)$ as described by \citet{woody2021model}. We see that the effects of the covariates is non-linear: generally, older individuals have poorer prognoses, as do individuals with high white blood cell counts, while higher deprivation scores are associated with worse outcomes up to a point before leveling off. This gives some explanation for the better predictive performance of the PH model over the Cox Linear and Weibull models.

\begin{figure}[t]
  \centering
  \includegraphics[width=1\textwidth]{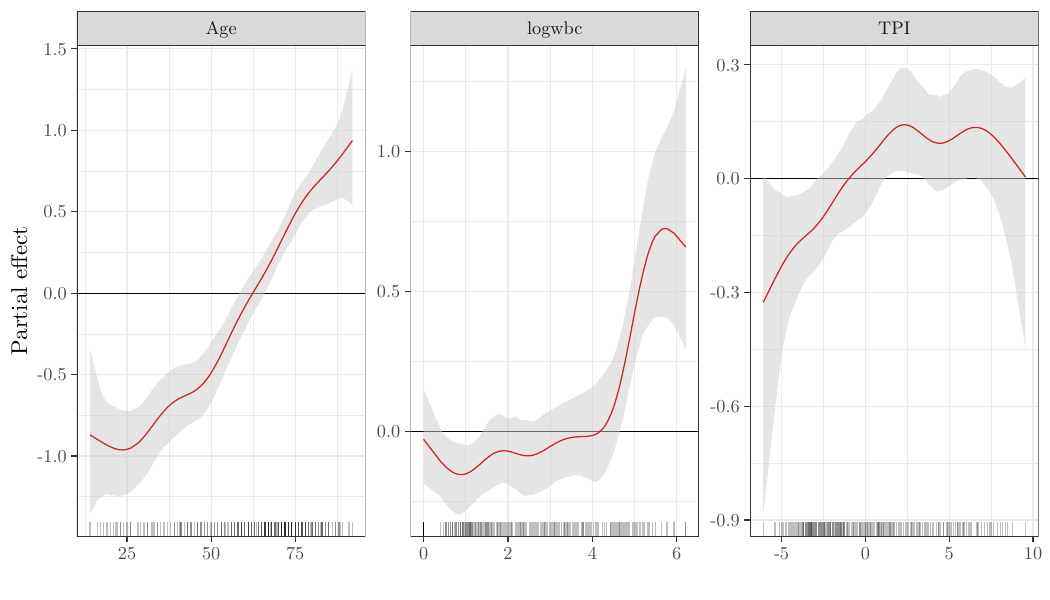}
  \caption{Posterior projection summaries for the partial effect of age, log of white blood cell count, and Townsend index.}
  \label{fig:addsum}
\end{figure}

\section{Discussion}
\label{sec:discussion}


In this work, we have proposed a unified Bayesian nonparametric framework for ordinal, survival, and density regression using the complementary log-log (\cloglog) link. This framework leverages the unique properties of the \cloglog\ link, particularly its connection to latent truncated exponential variables and Bayesian additive regression trees (BART), to enable computationally efficient posterior simulation. Compared to the analogous models using probit or logit links, our approach is both computationally simpler and more robust. For instance, the proposed ordinal regression framework avoids the need to impose constraints on the cutpoints, a common challenge in fitting cumulative link models. Similarly, our proportional hazards stick-breaking process (PHSBP) and its non-proportional extension (NPHSBP) require augmenting only a single latent variable per individual, making them more scalable than the alternatives, for example, the probit stick-breaking process \citep{rodriguez2011nonparametric}.

The practical utility of our methods is demonstrated through a series of simulated and real datasets. For example, in ordinal regression models, using depression level data from the Medical Expenditure Panel Survey (MEPS), we show that, apart from being easier to implement, the \cloglog\ link provides a better model fit compared to the probit link. Beyond computational and practical advantages, our work also provides theoretical insights into the role of the \cloglog\ link in Bayesian nonparametrics.  We establish posterior concentration rates for our ordinal regression models, showing that they achieve minimax-optimal rates up to logarithmic factors. Additionally, we connect the PHSBP to the weight-dependent Dirichlet process \citep[DDP,][]{maceachern2000dependent}, demonstrating that a special case of the PHSBP corresponds to a DDP, further highlighting its flexibility and generality.


Several promising directions for future research emerge from this work. First, while the \cloglog\ link has shown advantageous across multiple settings, exploring alternative asymmetric link functions could provide further flexibility. Another avenue is extending our framework to handle more complex data structures, such as longitudinal or spatial data, where the computational simplicity of our approach, combined with the flexibility of BART, makes it a natural candidate for extension. Additionally, further theoretical work could focus on deriving posterior concentration rates for the NPHSBP and survival models.  Also of interest would be extending our methods to accommodate the inclusion of additional forests, or sharing forests across different estimation tasks \citep{linero2020semiparametric}, which can be used for example in causal inference for estimating both average and heterogeneous treatment effects \citep{hahn2020bayesian}.

\bibliographystyle{apalike}
\bibliography{references}

\clearpage

\singlespacing

\renewcommand{\thesection}{S.\arabic{section}}
\renewcommand{\thesubsection}{S.\arabic{section}.\arabic{subsection}}
\renewcommand{\thetable}{S.\arabic{table}}
\renewcommand{\thefigure}{S.\arabic{figure}}
\renewcommand{\theequation}{S.\arabic{equation}}
\renewcommand{\thepage}{S.\arabic{page}}
\setcounter{page}{1}
\setcounter{section}{0}
\setcounter{subsection}{0}
\setcounter{table}{0}
\setcounter{figure}{0}
\setcounter{equation}{0}

\vspace*{2em}
\begin{center}
  \Large Web Appendix for 
A Unified Bayesian Nonparametric Framework for Ordinal, Survival, and Density
  Regression Using the Complementary Log-Log Link
\end{center}

\doublespacing

\section{Gibbs Sampler for the Proportional Hazards Ordinal Regression Model}
\label{sec:gibbs-sampler-ordinal}

We now describe the Gibbs sampling algorithm for the proportional hazards ordinal regression model. Let $\sM_t = \{\mu_{t\ell}\}_{\ell \in \Leaves(\Tree_t)}$ denote the leaf node parameters for tree $t$, and let $\theta = \{\gamma_j\}_{j=1}^{K-1}$ denote the remaining parameters. For observation $i$, we write $r(X_i) = \eta_i + \mu_{t\ell}$, where $\eta_i = \sum_{m \ne t} g(X_i; \Tree_m, \sM_m)$ denotes the contribution of all trees except tree $t$, and $g(X_i; \Tree_t, \sM_t) = \mu_{t\ell}$ if $X_i$ is associated with leaf node $\ell$ of tree $t$ (denoted $X_i \leadsto \ell$). Note that the likelihood of the data is given by (for fixed $t$):
\begin{align*}
  \prod_i \left[\{1 - \exp(-e^{\gamma_{Y_i} + r(X_i)})\} \prod_{j < Y_i} \exp(-e^{\gamma_j + r(X_i)})  \right].
\end{align*}
Hence, if we sample $Z_i \sim \TExp(e^{\gamma_{Y_i} + r(X_i)}, 0, 1)$  for $Y_i < K$ then after augmenting we have
\begin{align*}
  \prod \exp\left(1(Y_i < K)\{\gamma_{Y_i} + r(X_i) - Z_i e^{\gamma_{Y_i} + r(X_i)}\} - \sum_{j < Y_i} e^{\gamma_j + r(X_i)}\right).
\end{align*}
The algorithm alternates between three main steps: (i) updating the tree structures and leaf node parameters $(\Tree_t, \sM_t)$, (ii) updating latent variables $Z_i$ from truncated exponential distributions, and (iii) updating the threshold parameters $\gamma_j$. We detail each step below.

\paragraph{Updating Latent Variables.} For $i = 1,\ldots,n$, we update the latent variables according to
\begin{math}
[Z_i \mid \cdot] \sim \TExp(e^{\gamma_{Y_i} + r(X_i)}; 0, 1).
\end{math}

\paragraph{Updating Tree Parameters.} For the $t$th tree, we begin with the likelihood, which takes the form
\begin{align*}
L(\Tree_t, \sM_t) &= \prod_i \Pr(Y_i \mid \Tree_t, \sM_t; \Tree_{-t}, \sM_{-t}, \theta) \
\\&= \prod_{\ell \in \Leaves(\Tree_t)} \exp\left[\sum_{i: Y_i \neq K, X_i \leadsto \ell} \left(\gamma_{Y_i} + \eta_i + \mu_{t\ell} - Z_i e^{\gamma_{Y_i} + \eta_i + \mu_{t\ell}}\right) - \sum_{i: X_i \leadsto \ell} \sum_{j: j < Y_i} e^{\gamma_j + \eta_i + \mu_{t\ell}}\right].
\end{align*}
After dropping terms not involving $(\Tree_t, \sM_t)$, this reduces to
\begin{align*}
L(\Tree_t, \sM_t) &\propto \prod_{\ell \in \Leaves(\Tree_t)} \exp\left[\sum_{i: Y_i \neq K, X_i \leadsto \ell} \left(\mu_{t\ell} - Z_i e^{\gamma_{Y_i} + \eta_i + \mu_{t\ell}}\right) - \sum_{i: X_i \leadsto \ell} \sum_{j: j < Y_i} e^{\gamma_j + \eta_i + \mu_{t\ell}}\right] \
\\&= \prod_{\ell \in \Leaves(\Tree_t)} \exp\left[\mu_{t\ell} \sum_{i: Y_i \neq K, X_i \leadsto \ell} 1 - e^{\mu_{t\ell}}\left(\sum_{i: X_i \leadsto \ell} \sum_{j: j < Y_i} e^{\gamma_j + \eta_i} + \sum_{i: Y_i \neq K, X_i \leadsto \ell} Z_i e^{\gamma_{Y_i} + \eta_i}\right)\right] \
\\&= \prod_{\ell \in \Leaves(\Tree_t)} \exp(\mu_{t\ell} A_\ell - e^{\mu_{t\ell}} B_\ell),
\end{align*}
where we define
\begin{align*}
A_\ell &= \sum_{i: Y_i \neq K, X_i \leadsto \ell} 1 \qquad \text{and} \qquad
B_\ell = \sum_{i: X_i \leadsto \ell} \sum_{j: j < Y_i} e^{\gamma_j + \eta_i} + \sum_{i: Y_i \neq K, X_i \leadsto \ell} Z_i e^{\gamma_{Y_i} + \eta_i}.
\end{align*}
Under the $\log \Gam(a, b)$ prior for the leaf node parameters, the integrated likelihood for a given tree structure $\Tree_t$ is
\begin{align*}
L(\Tree_t) &= \prod_{\ell \in \Leaves(\Tree_t)} \int \exp(\mu_{t\ell} A_\ell - e^{\mu_{t\ell}} B_\ell) \frac{b^a}{\Gamma(a)} \exp(a\mu_{t\ell} - be^{\mu_{t\ell}})  d\mu_{t\ell} \
\\&= \prod_{\ell \in \Leaves(\Tree_t)} \frac{b^a}{\Gamma(a)} \int \exp\{(a + A_\ell)\mu_{t\ell} - (b + B_\ell)e^{\mu_{t\ell}}\}  d\mu_{t\ell} \
\\&= \prod_{\ell \in \Leaves(\Tree_t)} \frac{b^a}{\Gamma(a)} \cdot \frac{\Gamma(a + A_\ell)}{(b + B_\ell)^{(a + A_\ell)}}.
\end{align*}
Noting the kernel of the $\log \Gam(a + A_{\ell}, b + B_{\ell})$ distribution in the integral, this also implies that
the full conditional distribution of $\mu_{t\ell}$ is
\begin{align*}
[\mu_{t\ell} \mid \text{everything else}] \sim \log \Gam(a + A_\ell, b + B_\ell).
\end{align*}

\paragraph{Updating Threshold Parameters.}

From the joint likelihood, dropping the terms not involving $\gamma_j$, the likelihood of each $\gamma_j$ is given by \begin{align*}
\pi(\gamma_j \mid \text{everything else}) &\propto \prod_{\ell \in \Leaves(\Tree_t)} \exp\left[\sum_{i: Y_i \neq K, X_i \leadsto \ell} \left(\gamma_{Y_i} + \eta_i + \mu_{t\ell} - Z_ie^{\gamma_{Y_i} + \eta_i + \mu_{t\ell}}\right) \right. \
\\&\quad \left. - \sum_{i: X_i \leadsto \ell} \sum_{j: j < Y_i} e^{\gamma_{j} + \eta_i + \mu_{t\ell}}\right] \cdot \pi(\gamma) \
\\&\propto \exp\left[\gamma_j \left\{a_\gamma + \sum_{i: Y_i = j} 1\right\} - e^{\gamma_j}\left\{b_\gamma + \sum_{i: Y_i = j} Z_ie^{r(X_i)} + \sum_{i: Y_i > j} e^{r(X_i)}\right\}\right].
\end{align*}
Defining
\begin{align*}
A_{\gamma_j} = \sum_{i: Y_i = j} 1 \qquad\text{and}\qquad
B_{\gamma_j} = \sum_{i: Y_i = j} Z_ie^{r(X_i)} + \sum_{i: Y_i > j} e^{r(X_i)}.
\end{align*}
we see that the full conditional distribution for $\gamma_j$ takes the form of a log-gamma distribution:
\begin{align*}
[\gamma_j \mid \text{everything else}] \sim \log \Gam(a_\gamma + A_{\gamma_j}, b_\gamma + B_{\gamma_j}).
\end{align*}

\section{Gibbs Sampler for the Non-Proportional Hazards Ordinal Regression Model}
\label{sec:gibbs-sampler-ordinal}

The key difference in the non-proportional hazards model is that the nonparametric function $r(x,k)$ depends on both the covariates and the category. Let $\sM_t = \{\mu_{t\ell}\}_{\ell \in \Leaves(\Tree_t)}$ denote the leaf node parameters for tree $t$, and $\theta = \{\gamma_j\}_{j=1}^{K-1}$ denote the threshold parameters. For observation $i$, we write $r(X_i,k) = \eta_i(k) + \mu_{t\ell}$, where $\eta_i(k) = \sum_{m \ne t} g(X_i,k; \Tree_m, \sM_m)$ denotes the contribution of all trees except tree $t$ for category $k$.

\paragraph{Updating Latent Variables.} The update for the latent variables remains similar to the proportional hazards case, but now incorporates the category-specific predictions:
\begin{align*}
  [Z_i \mid \text{everything else}] \sim \TExp(e^{\gamma_{Y_i} + r(X_i,Y_i)}; 0, 1),
\end{align*}
and we only update $Z_i$ when $Y_i < K$.

\paragraph{Updating Tree Parameters.} After augmenting the latent exponential random variables, the likelihood takes the form
\begin{align*}
  &\prod_i \exp\left\{ 1(Y_i < K) \{\gamma_{Y_i} + r(X_i, Y_i) - Z_i \, e^{\gamma_{Y_i} - r(X_i, Y_i)}\} - \sum_{j < Y_i} e^{\gamma_j + r(X_i, j)} \right\}.
\end{align*}
For simplicity, let us define for $j = 1,\ldots,K-1$ the quantities $Z_{ij} = 1$ if $Y_i > j$, $Z_{ij} = Z_i$ if $Y_i = j$ and $Z_{ij} = 0$ if $Y_i > j$. Then we can rewrite this expression as
\begin{align*}
  \prod_i \exp\left\{ \sum_{j=1}^{K-1} {1(Y_i = j)\{\gamma_j + r(X_i,j)\} - Z_{ij} e^{\gamma_j + r(X_i,j)}} \right\}.
\end{align*}
Breaking the likelihood up over the leaf nodes, we can write this as
\begin{align*}
  &\prod_{\ell \in \Leaves(\Tree_t)} \exp\left\{ \sum_{j=1}^{K-1} \sum_{i : (X_i, j) \leadsto \ell} {1(Y_i = j)\{\gamma_j + r(X_i,j)\} - Z_{ij} e^{\gamma_j + r(X_i,j)}} \right\}.
  \\&= \prod_{\ell \in \Leaves(\Tree_t)} \exp\left\{ \sum_{j=1}^{K-1} \sum_{i : (X_i, j) \leadsto \ell} 1(Y_i = j)\{\gamma_j + \eta_i(j) + \mu_{t\ell}\} - Z_{ij} e^{\gamma_j + \eta_i(j)} e^{\mu_{t\ell}} \right\}.
  \\&\propto \prod_{\ell \in \Leaves(\Tree_t)} \exp\left\{ A_{\ell} \mu_{t\ell} - B_{\ell} e^{\mu_{t\ell}}\right\},
\end{align*}
where $A_{t\ell} = \sum_{j = 1}^{K-1} \sum_{i : (X_i, j) \leadsto \ell} 1(Y_i = j)$ is the number of $(X_i, Y_i)$ pairs associated to $\mu_{t\ell}$ (excluding $Y_i = K$) and $B_{t\ell} = \sum_{j = 1}^{K - 1} \sum_{i : (X_i, j) \leadsto \ell} Z_{ij} e^{\gamma_j + \eta_i(j)}$. As before, this leads to the marginal likelihood of $\Tree_t$ being proportional to
\begin{align*}
  \prod_{\ell \in \Leaves(\Tree_t)} \frac{b^a}{\Gamma(a)} \cdot \frac{\Gamma(a + A_\ell)}{(b + B_\ell)^{(a + A_\ell)}}.
\end{align*}
and the full conditional $\mu_{t\ell} \sim \log \Gam(a + A_{\ell}, b + B_{\ell})$.

\paragraph{Updating Threshold Parameters.} The full conditional distribution of $\gamma_j$ can be derived similarly to the proportional hazards case. From the augmented likelihood, the terms involving $\gamma_j$ are, for $j = 1, \ldots, K - 1$, given by
\begin{align*}
  \exp\left[\gamma_j \sum_{i: Y_i = j} 1 - e^{\gamma_j} \left(\sum_{i: Y_i = j} Z_i e^{r(X_i,j)} + \sum_{i: Y_i > j} e^{r(X_i,j)}\right)\right].
\end{align*}
Combining this with the log-gamma prior $\gamma_j \sim \log \Gam(a_\gamma, b_\gamma)$, the full conditional distribution becomes:
\begin{align*}
  [\gamma_j \mid \text{everything else}] &\sim \log \Gam\left(a_\gamma + \sum_{i: Y_i = j} 1, \: b_\gamma + \sum_{i: Y_i = j} Z_i e^{r(X_i,j)} + \sum_{i: Y_i > j} e^{r(X_i,j)}\right).
\end{align*}

\section{Gibbs Sampler for Density Regression}
\label{sec:gibbs-density-regression}

The Gibbs sampler for the stick-breaking models follows the same structure as for the ordinal regression models. The only new steps are that we need to update the latent cluster indicators and we need to update the parameters of the mixture components.

\paragraph{Update cluster assignments.} For each observation $i$, sample a latent cluster assignment $C_i$ from its full conditional distribution. The full conditional distribution of $C_i$ is given by
\begin{align*}
  \Pr(C_i = k \mid \text{everything else}) \propto \varpi_k(X_i) \, \Normal\{Y_i \mid \mu_k + h(X_i), \sigma^2_k\}.
\end{align*}

\paragraph{Update the regression trees.} The tree parameters are updated using a weighted regression update as described, for example, by \citet{pratola2020heteroscedastic}. For the mean function $h(x)$, we update the trees using the residuals $Y_i - \mu_{C_i}$ with weights $1/\sigma^2_{C_i}$, where $C_i$ is the cluster assignment and $(\mu_k, \sigma^2_k)$ are the parameters of mixture component $k$. The variance associated with this regression update is fixed at $\sigma^2 = 1$.

\paragraph{Update the stick-breaking weights.} Treat the $C_i$'s as ordinal data, which is then used to perform the updates for $\varpi_k(\cdot)$ using either the PHOBART or NPHOBART MCMC algorithms described previously, depending on whether one wants to fit the PHSBP or NPHSBP mixture model.

\paragraph{Update the mixture component parameters and hyperparameters.} Update the $\mu_k$'s and $\sigma_k$'s by sampling from the full conditionals based on the observations $Y_i - h(X_i) \sim \Normal(\mu_k, \sigma^2_k)$ when $C_i = k$. Then sample the parameters of the base measure $(\mu_k, \sigma^{-2}_k) \sim \Normal(\mu_0, \sigma^2_0) \, \Gam(a_\sigma, b_\sigma)$ from their full conditionals.

\section{Gibbs Sampler for Survival Models}
\label{sec:gibbs-survival}

We begin with the proportional hazards survival model. Recall that the likelihood of this model is given by
\begin{align*}
  \prod_i
  \exp\left\{\delta_i \, \log \lambda_{B_i} + \delta_i r(X_i) - 
  e^{r(X_i)} \sum_{b = 1}^{B_i - 1} \lambda_b (t_b - t_{b-1}) - 
  e^{r(X_i)} \lambda_{B_i}(Y_i - t_{B_i - 1})\right\}.
\end{align*}

\paragraph{Update for the $\lambda_b$'s.}
Under the $\Gam(a_\lambda, b_\lambda)$ prior for $\lambda_b$, the relevant terms in the likelihood are:
\begin{align*}
\exp\left\{\sum_{i: B_i = b} \delta_i \log \lambda_b - \lambda_b \sum_{i: Y_i > t_b} e^{r(X_i)}(t_b - t_{b-1}) - \lambda_b \sum_{i: B_i = b} e^{r(X_i)}(Y_i - t_{b-1})\right\}.
\end{align*}
Combined with the gamma prior, this gives a conjugate update:
\begin{align*}
[\lambda_b \mid \text{everything else}] \sim \Gam\left(a_\lambda + \sum_{i: B_i = b} \delta_i, , b_\lambda + \sum_{i: Y_i > t_b} e^{r(X_i)}(t_b - t_{b-1}) + \sum_{i: B_i = b} e^{r(X_i)}(Y_i - t_{b-1})\right).
\end{align*}

\paragraph{Update for the trees.} Letting $\eta_i = \sum_{k : k \ne t} g(X_i; \Tree_t, \sM_t)$, the likelihood for $(\Tree_t, \sM_t)$ is proportional to
\begin{align*}
    \prod_{\ell \in \Leaves(\Tree_t)}
    \prod_{i : X_i \leadsto \ell}
    \exp\left\{\delta_i \mu_{t\ell} - e^{\mu_{t\ell}} e^{\eta_i} \sum_{b = 1}^{B_i - 1} \lambda_b (t_b - t_{b-1}) - e^{\mu_{t\ell}} e^{\eta_i} \lambda_{B_i}(Y_i - t_{B_i - 1})\right\}
\end{align*}
Therefore, as in previous calculations, the marginal likelihood of $\Tree_t$ is given by
\begin{align*}
  \prod_{\ell \in \Leaves(\Tree_t)} \frac{b^a}{\Gamma(a)} \times \frac{\Gamma(a + A_\ell)}{(b + B_\ell)^{a + A_\ell}}
\end{align*}
where $A_\ell = \sum_{i: X_i \leadsto \ell} \delta_i$ and $B_\ell = \sum_{i: X_i \leadsto \ell} \left[\sum_{b = 1}^{B_i - 1} \lambda_b (t_b - t_{b-1}) + \lambda_{B_i}(Y_i - t_{B_i - 1})\right]e^{\eta_i}$. The full conditional distribution of the leaf parameters is therefore $\mu_{t\ell} \sim \log \Gam(a + A_\ell, b + B_\ell)$.

For the non-proportional hazards model, we need to modify the updates to account for the interaction between covariates and time bins. The likelihood becomes:
\begin{align*}
\prod_i \exp\left\{\delta_i \log \lambda_{B_i} + \delta_i r(X_i, B_i) - \sum_{b = 1}^{B_i - 1} \lambda_b e^{r(X_i, b)}(t_b - t_{b-1}) - \lambda_{B_i} e^{r(X_i, B_i)}(Y_i - t_{B_i - 1})\right\}
\end{align*}
The $\lambda_b$ update becomes:
\begin{align*}
[\lambda_b \mid \text{everything else}] \sim \Gam\left(a_\lambda + \sum_{i: B_i = b} \delta_i, b_\lambda + \sum_{i: Y_i >  t_b} e^{r(X_i,b)}(t_b - t_{b-1}) + \sum_{i: B_i = b} e^{r(X_i,b)}(Y_i - t_{b-1})\right).
\end{align*}
For the trees, let $\eta_{ib} = \sum_{k \ne t} g(X_i, b ; \Tree_k, \sM_k)$, and define further the quantities
\begin{align*}
  Z_{ib} = 1(Y_i \ge t_b) (t_b - t_{b-1}) + 1(t_{b-1} \le Y_i < t_b) (Y_i - t_{b-1}).
\end{align*}
Then, by essentially the same arguments as above, the marginal likelihood and full conditionals are of the same form but now with
\begin{align*}
  A_\ell = \sum_{(i,b) \leadsto \ell} \delta_i \, 1(B_i = b)
  \qquad \text{and} \qquad
  B_\ell = \sum_{(i,b) \leadsto \ell} Z_{ib} \lambda_b e^{\eta_{ib}}.
\end{align*}

\section{Proof of Theorem~\ref{thm:posterior-concentration}}
\label{sec:proof-of-theorem}

The overall strategy for establishing Theorem~\ref{thm:posterior-concentration} is to verify the sufficient conditions given by \citet{ghosal2000convergence}. This involves establishing both that the prior assigns sufficiently high mass to shrinking neighborhoods around the true data generating mechanism and that the prior assigns sufficiently high mass to a low-entropy sieve on the model space.

We start by introducing some notation. Let $\Pi = \Pi_r \times \Pi_\gamma$ denote the prior distribution of the model parameters $\theta = (r, \gamma)$, with $\Pi_r$ and $\Pi_\gamma$ denoting the respective distributions of $r$ and $\gamma = (\gamma_1, \ldots, \gamma_{K-1})$. We denote the observed data as $\Data_n = (Y_i, X_i: i = 1, \dots, n)$, and write $\Pi(\cdot \mid \Data_n)$ to denote the posterior distribution of the model parameters given the data $\Data_n$. We write $a_n \lesssim b_n$ if there exists a global constant $c > 0$ (independent of $n$ and $P$) such that $a_n \leq cb_n$ for sufficiently large $n$, and $a_n \asymp b_n$ if $a_n \lesssim b_n$ and $b_n \lesssim a_n$. We use $C$ to denote a generic positive constant (independent of $n$ and $p$) that may change from line to line. We assume that the prior on $\gamma_j$ is a log-gamma distribution, i.e., $\gamma_j \stackrel{\text{iid}}{\sim} \log \Gam(a_\gamma, b_\gamma)$ for $j = 1, \ldots, K - 1$. Throughout we let $\|r\|_\infty = \sup_{x \in \scriptX} |r(x)|$, with $\scriptX = [0,1]^P$, and $\|\gamma\|_\infty = \max_{j = 1, \ldots, K - 1} |\gamma_j|$.  The following conditions on $\Pi_r$ are adapted from \citet{orlandi2021density}, which are themselves variations of the assumptions proposed in \citet{jeong2023art}.

\paragraph{Condition P (on $\Pi$).} Let $S \subseteq \{1, \ldots, P\}$ denote the coordinates of $x$ at which the trees $\Tree_1, \ldots, \Tree_T$ split.
\begin{enumerate}
    \item[(P1)] \emph{Prior on the support set.} The support set $S$ of $r$ has prior $\Pi_S(S) = P^{-D} \pi_D(D)$, where $D \equiv |S|$ and $\pi_D(d)$ is an exponentially decaying prior satisfying
    \[
    a_1P^{-a_3} \pi_D(d-1) \leq \pi_D(d) \leq a_2P^{-a_4}\pi_D(d-1)
    \]
    for some positive constants $a_1, \ldots, a_4$ and $d = 1, \ldots, P$.

    \item[(P2)] \emph{Tree prior.} Given $S$, each tree $\Tree_t$, $t = 1, \ldots, T$ is assigned a heterogeneous Galton-Watson branching process prior with splitting proportion $q(d) = \nu^d$ for some $\nu \in (0, 1/2)$.

    \item[(P3)] \emph{Leaf node parameters prior.} The leaf node parameters $\mu_{tl}$ of $\Tree_t$ are given independent $\text{Normal}(0, \sigma^2_\mu)$ priors.

    \item[(P4)] \emph{Split point selection.} Splits in the tree ensemble can occur only at a number $b_n$ of candidate split points $\mathcal{Z}_n \subseteq [0, 1]^P$, which are selected from uniformly. Additionally, $\log b_n \lesssim \log n$.

    \item[(P5)] \emph{Approximation capability of tree.} For each $n$ there exists a decision tree $\Tree$ and leaf node values $\Leaves$ built from the candidate split-points in $\mathcal{Z}_n$ such that the regression tree $r^\star(x) = g(x; \Tree, \Leaves)$ satisfies $\|r_0 - r^\star\|_\infty \lesssim (\log n/n)^{\alpha/(2\alpha+D_0)}$ where $r_0$ is the true $D_0$-sparse regression function satisfying Condition F.

    \item[(P6)] The parameters $r$ and $\gamma$ are a-priori independent.
\end{enumerate}

Condition P5 should be interpreted as a condition on the candidate split-points, and it holds very generally given that $r_0$ satisfies Condition F; see \citet{jeong2023art} for specific conditions.

We make use of the following Lemma \ref{lem:s1} and \ref{lem:s3} stated in \citet{orlandi2021density}, which can be extracted by analyzing the proofs of Lemma 5 –- 7 of \citet{jeong2023art}, and Lemma \ref{lem:s5} from \citet{linero2024generalized}. Lemma \ref{lem:s1} is a concentration result for the BART prior. We next provide a concentration result for log-gamma priors in Lemma \ref{lem:s2}. These two prior concentration results will be useful in showing that the prior $\Pi$ assigns adequate probability mass to neighborhoods of the true data-generating process.

\begin{lemma}
  \label{lem:s1}
  Suppose that Condition F and Condition P hold and let $\tilde \varepsilon_n \propto (\log n/n)^{\alpha/(2\alpha+D_0)}$. Then for sufficiently large $n$ we have
  \[
  \Pi_r(\|r - r_0\|_\infty \leq \tilde \varepsilon_n \mid S = S_0) \geq e^{-\tilde C_{SB}n{\tilde \varepsilon_n}^2},
  \]
  where $\tilde C_{SB}$ is a positive constant, where $S_0$ denotes the set of relevant predictors in $r_0$.
\end{lemma}

As argued by \citet{rovckova2020posterior}, (P1) implies  $-\log \Pi_S(S = S_0) \lesssim D_0 \log(P + 1) \lesssim n\tilde \varepsilon_n^2$. Thus, using total probability,
\begin{align*}
  \Pi_r(\|r - r_0\|_\infty \leq \tilde \varepsilon_n) \geq \Pi_r(\|r - r_0\|_\infty \leq \tilde \varepsilon_n \mid S = S_0) \Pi_S(S = S_0) \ge e^{-C_{SB}n{\tilde \varepsilon_n}^2},
\end{align*}
where $C_{SB}$ is a positive constant depending on $(D_0, \log P, \tilde C_{SB})$.

\begin{lemma}
  \label{lem:s2}
  Let $\bar \varepsilon_n = (\log n / n)^{1/2}$. Then for sufficiently large $n$ we have
  \[
    \Pi_\gamma(\|\gamma - \gamma_0\|_\infty \leq \bar \varepsilon_n) \geq e^{-C_{Gam}n{\bar\varepsilon_n}^2},
  \]
  where $C_{Gam}$ is a positive constant.
\end{lemma}

Lemma \ref{lem:s3} and \ref{lem:s4} provide bounds on the metric entropy of certain sieves on the model spaces for $r$ and $\gamma$.

\begin{lemma}
  \label{lem:s3}
  Let $\mathcal{F}^\star_n$ denote the collection of decision tree ensembles with $T$ trees that
  (i) split on no more than $d$ variables, (ii) have at most $L$ leaf nodes per tree, (iii) have at most $b_n$ candidate split points, and (iv) satisfy $\sup_{t,l} |\mu_{tl}| \leq U$.
  Then
  \[
    \log N(\mathcal{F}^\star_n, \varepsilon, \|\cdot\|_\infty) \lesssim d \log P + L \log \frac{d^T b^T_n L U}{\varepsilon},
  \]
  where $N(\mathcal{F}^\star_n, \varepsilon, \rho)$ denotes the number of $\varepsilon$-balls required to cover $\mathcal{F}^\star_n$ with respect to a distance $\rho$.
\end{lemma}

\begin{lemma}
  \label{lem:s4}
  Let $\mathcal{G}^\star_n$ denote the collection of $\gamma$ vectors with $\|\gamma\|_\infty \leq G$. Then
  \[
    N(\mathcal{G}^\star_n, \varepsilon, \|\cdot\|_\infty)
    \le \left\lceil \frac{G}{\varepsilon} \right\rceil^{K-1}
    \lesssim \left(\frac{G}{\varepsilon}\right)^{K-1}
  \]
  where $N(\mathcal{G}^\star_n, \varepsilon, \rho)$ denotes the number of $\varepsilon$-balls required to cover $\mathcal{G}^\star_n$ with respect to a distance $\rho$.
\end{lemma}


The next two lemmas demonstrate that the prior assigns exponentially small probability outside the set $\mathcal{F}_n = \{\theta : r \in \mathcal F^\star_n, \gamma \in \mathcal G^\star_n\}$. This low-entropy high-mass property of sieve is required to apply results of \citet{ghosal2000convergence} in order to derive posterior contraction rates.

\begin{lemma}
  \label{lem:s5}
  Let $\mathcal{F}^\star_n$ be defined as in Lemma \ref{lem:s3} and suppose Condition P holds. Then
  \[
    \Pi(r \notin \mathcal{F}^\star_n) \leq e^{-c_D d \log P} + T e^{-C_L L \log L} + LT e^{-U^2/(2\sigma^2_\mu)},
  \]
  for some positive constants $C_L$, $c_D$ and sufficiently large $n$.
\end{lemma}

\begin{lemma}
  \label{lem:s6}
  Let $\mathcal G^\star_n$ be as defined in Lemma~\ref{lem:s4} and suppose $\gamma_j \sim \log \Gam(a_\gamma, b_\gamma)$. Then
  \begin{align*}
    \Pi(\gamma \notin \mathcal G^\star_n) \le K_\gamma e^{-C_\gamma G}
  \end{align*}
  for some $C_\gamma > 0$ and $K_\gamma > 0$.
\end{lemma}

Using Lemma \ref{lem:s5} and \ref{lem:s6}, we have
\[
  \Pi(\theta \notin \mathcal{F}_n)
  \leq e^{-c_D d \log P} + T e^{-C_L L \log L} + LT e^{-U^2/(2\sigma^2_\mu)}
    + K_\gamma e^{-C_\gamma G}.
\]
We next require the following Lemma~\ref{lem:s7} involving the density of Gumbel random variables $Z_i$. This lemma is a consequence of Corollary 1 in \citet{linero2024generalized}.

\begin{lemma}
  \label{lem:s7}
  Suppose $Z = \mu + \epsilon_i$ with $\epsilon_i \sim h(\epsilon) = \exp(\epsilon - e^\epsilon)$ and let $f$ denote the density of $Z$. Then the family $\{f(z \mid \mu) : \mu \in \R\}$ satisfies Condition A1 of \citet{linero2024generalized}, i.e.,
  \begin{align*}
    \max_{w = 1, 2} \int_{-\infty}^\infty f(z \mid \mu)
      \left( \log \frac{f(z \mid \mu)}{f(z \mid \mu + \Delta)} \right)^w \ d\mu
    \le C \Delta^2
  \end{align*}
  for all $|\Delta| \le 1$, where $C$ is a universal constant.
\end{lemma}

Lastly, we require an equivalent latent variable formulation of the proportional hazards model, which is given by Lemma~\ref{lem:s8}.

\begin{lemma}
  \label{lem:s8}
  Let $Z_k = -(\gamma_k + r) + \epsilon_k$ where $\epsilon_k \sim \log \Gam(1,1)$. Define $Y = \min\{k : Z_k < 0\}$. Then $\Pr(Y > k \mid \gamma, r) = \exp(-e^{c_k + r})$ where $c_k = \log \sum_{j \le k} e^{\gamma_j}$.
\end{lemma}


We prove Theorem \ref{thm:posterior-concentration} by verifying the following sufficient conditions to establish a posterior contraction rate of $\epsilon_n$ (see, e.g., \citet{li2023adaptive}). Let $p_\theta = p_\theta(y \mid x)$ denote the likelihood function of the model, and similarly, $p_0 = p_{\theta_0}(y \mid x)$ denote the likelihood of the true model. We define the \emph{integrated Kullback-Leibler} neighborhoods of $\theta_0$ with radii $\epsilon_n$ as
\begin{align}
  \label{eq:ikl}
  \text{KL}(\theta_0, \epsilon_n) = \Big\{\theta = (r, \gamma_j: j = 1, \dots, K-1): & \int_{\mathcal{X}} KL(p_{\theta_0} \| p_\theta) F_X(dx) \leq \epsilon_n^2, \nonumber \\
  & \quad \int_{\mathcal{X}} V(p_{\theta_0}\|p_\theta) F_X(dx) \leq \epsilon_n^2\Big\},
\end{align}
where $KL(p_{\theta_0}\| p_\theta) = \sum_y p_{\theta_0} \log \frac{p_{\theta_0}}{p_\theta}$ and $V(p_{\theta_0}\| p_\theta) = \sum_y p_{\theta_0}  \left(\log \frac{p_{\theta_0}}{p_\theta}  \right)^2$.
We attain a posterior convergence rate of $\epsilon_n$ with respect to the \emph{integrated Hellinger distance} defined as
$$
  H(\theta_0, \theta) = \left[\int \sum_y \left\{\sqrt{p_{\theta_0}(y|x)} - \sqrt{p_\theta(y|x)} \right\}^2  F_X(dx)\right]^{1/2},
$$
if there exist positive constants $c_1, c_2, c_3, c_4$ and sieve $\mathcal{H}_n$ of conditional densities $\{p_\theta(y|x): \theta \in \mathcal{F}_n\}$ such that
\begin{align*}
  &\text{(B1)} \emph{ Positive prior mass or thickness. } \Pi\{\theta \in \text{KL}(\theta_0, \epsilon_n)\} \geq c_1 \exp\{-c_2 n \epsilon_n^2\} \\
  &\text{(B2)} \emph{ Low entropy sieve. }  \log N(\mathcal{H}_n, \epsilon_n, H) \leq c_3 n\epsilon_n^2 \\
  &\text{(B3)} \emph{ High mass sieve. } \Pi\{\theta \notin \mathcal{F}_n\} \leq c_4 \exp\{-(c_2 + 4)n\epsilon_n^2\}.
\end{align*}
Additionally, since the integrated total variation distance is upper-bounded by the integrated Hellinger distance, this suffices to prove Theorem \ref{thm:posterior-concentration}.

\paragraph{Condition B1.}
We begin by showing that, for sufficiently small $\Delta$, there exists a positive constant $C$ (depending on $r_0$ and $\gamma_0$) such that the integrated Kullback-Leibler neighborhood of radius $\sqrt{C} \Delta$ is contained in $\{p_\theta : \|\gamma - \gamma_0\|_\infty \vee \|r - r_0\|_\infty \le \Delta\}$. Let $Z_{ik} = -(\gamma_k + r(X_i)) + \epsilon_{ik}$ and let $Y_i = \min\{k : Z_{ik} < 0\}$. Using Lemma \ref{lem:s8} and the data processing inequality for $f$-divergences, the integrated KL-divergence between $p_0$ and $p_\theta$ is bounded by,
\begin{align}
  \label{eq:dpi}
  \int KL(p_0 \| p_\theta) \ F_X(dx) \leq
  \int KL(f_0 \| f_\theta) \ F_X(dx) =
  \sum_{k=1}^{K-1} \int KL(f_{0, k} \| f_{\theta, k}) \ F_X(dx),
\end{align}
where $f_{0,k}$ and $f_{\theta, k}$ are the conditional densities of the $Z_{ik}$ under $\theta_0$ and $\theta$, respectively, and $f_0 = \prod_{k=1}^{K-1} f_{0,k}$ and $f_\theta = \prod_{k=1}^{K-1} f_{\theta, k}$. The last equality in (\ref{eq:dpi}) is due to independence of the $Z_{ik}$'s and KL-divergence being additive over product measures. Now, using Lemma~\ref{lem:s7} for $Z_{ik}$, we upper bound $KL(f_{0, k} \| f_{\theta, k}) = KL(f(\cdot \mid r_0(x) + \gamma_{0k}) \| f(\cdot \mid r(x) + \gamma_k)) \leq C \{r(x) + \gamma_k - r_0(x) - \gamma_{0k}\}^2.$ Thus, integrated KL is upper bounded as
\begin{align}
  \label{eq:kl}
  \int KL(p_0 \| p_\theta) F_X(dx) \leq C \int \sum_{k=1}^{K-1} \left\{|r(x) - r_0(x)| + |\gamma_k  - \gamma_{0k}|\right\}^2 F_X(dx) \leq C_K \Delta^2,
\end{align}
where  $\Delta = \|r-r_0\|_\infty \vee \|\gamma - \gamma_0\|_\infty$ and $C_K$ is a constant depending only on $K$.

Next, we bound $V(p_0 \| p_\theta)$ for fixed $x$. First, we note that trivially
\begin{align*}
  V(p_0 \| p_\theta) = \sum_{k = 1}^K \{\log p_\theta(k \mid x) - \log p_0(k \mid x)\}^2 \, p_0(k \mid x)
\end{align*}
Now, define
\begin{align*}
  A_k &= [\log \{1 - \exp(-e^{r(x) + \gamma_k})\} - \log \{1 - \exp(-e^{r_0(x) + \gamma_{0k}})\}] 1(k < K) \\
  B_k &= e^{r_0(x) + \gamma_{0k}} - e^{r(x) + \gamma_k}.
\end{align*}
Then we can bound $V(p_0 \| p_\theta)$ using Cauchy-Schwarz as
\begin{align*}
  \sum_{k = 1}^K p_0(k \mid x) \left\{ A_k + \sum_{j < k} B_j \right\}^2
  \le K \sum_{k = 1}^K p_0(k \mid x) (A_k^2 + \sum_{j < k} B_j^2).
\end{align*}
By Taylor expansion, it is easy to show that for $\Delta \le \log 2$ we have
\begin{align*}
  |B_k| &\le 2 e^{\|r_0\|_\infty + \|\gamma\|_\infty} \Delta \\
  |A_k| &\le \left( \sup_{x, \delta \le \log 2} \frac{\exp(r_0(x) + \gamma_{0k} + \delta - e^{r_0(x) + \gamma_{0k} + \delta})}{1 - \exp(-e^{r_0(x) + \gamma_{0k} + \delta})} \right) \Delta.
\end{align*}
It follows from this that $\int V(p_0 \| p_\theta) \ F_X(dx) \le C_V \, \Delta^2$ for some $C_V$ determined by $(K, \|r_0\|_\infty, \|\gamma_0\|_\infty)$, provided $\Delta$ is sufficiently small.

Using this, we lower bound $\Pi\{\theta \in \text{KL}(\theta_0, \epsilon_n)\}$ for $n$ sufficiently large. First, using the argument above, for $n$ sufficiently large we have $\text{KL}(\theta_0, \epsilon_n) \supseteq \{\theta: \|r-r_0\|_\infty \leq \epsilon_n / \sqrt{C},  \|\gamma - \gamma_0\|_\infty \leq \epsilon_n / \sqrt{C}\}$ for some $C$ depending on $(K, r_0, \gamma_0)$. Thus
\begin{align*}
  \Pi\{\theta \in \text{ KL}(\theta_0, \epsilon_n)\} \geq \Pi_r\{\|r-r_0\|_\infty \leq \epsilon_n / \sqrt{C}\} \times  \Pi_\gamma\{\|\gamma - \gamma_0\|_\infty \leq \epsilon_n / \sqrt{C}\}.
\end{align*}
Using Lemma \ref{lem:s1} and \ref{lem:s2}, with $\epsilon_n = \max\{\tilde \varepsilon_n, \bar \varepsilon_n\}$ and for sufficiently large $n$,
\begin{align}
\label{eq:b1}
  \Pi\{\theta \in \text{ KL}(\theta_0, \epsilon_n)\} \geq e^{-C_{SB} n (\epsilon_n / \sqrt{C})^2} \times e^{-C_{Gam} n (\epsilon_n / \sqrt{C})^2} = c_1e^{- c_2 n \epsilon_n^2},
\end{align}
with $c_1 = 1$ and $c_2 = \frac{C_{SB} + C_{Gam}}{C}$. This completes the verification of Condition B1.

\paragraph{Conditions B2 and B3.}
Let $\mathcal H_n = \{p_\theta : \theta \in \mathcal F_n\}$ where $\mathcal F_n$ is determined by to-be-specified choices of $(d, L, U, G)$. As shown in the verification of Condition B1, we can bound the Hellinger distance in terms of Kullback-Leibler as
$$
  H^2(\theta', \theta) \leq (1/2) \int KL(p_{\theta'} \| p_\theta) F_X(dx) \leq C_K \epsilon^2,
$$
for some constant depending only on $K$ and where $\epsilon = \|r - r'\|_\infty \vee  \|\gamma - \gamma'\|_\infty$, provided that $\epsilon$ is sufficiently small. Therefore, by Lemma \ref{lem:s3} and \ref{lem:s4}, for large enough $n$ we have
\begin{align}
\label{eq:b2}
  \log N(\mathcal{H}_n, \epsilon_n, H)
  &\leq \log N(\mathcal{F}_n, \epsilon_n / \sqrt{C_K}, \|\cdot\|_\infty)
  \\&\lesssim d \log P + L \log \frac{d^T b^T_n L U \sqrt{C}}{\epsilon_n} + (K-1) \log \frac{G \sqrt{C}}{\epsilon_n}.
\end{align}
From (P4), we have $\log b_n \lesssim \log n$. Fix a large  $\kappa > 0$, and define $d = \frac{\kappa n \epsilon_n^2}{\log P}, L = \left\lfloor\frac{\kappa n \epsilon_n^2}{\log n}\right\rfloor, U^2 = \kappa n \epsilon_n^2, G = \kappa n \epsilon_n^2$. For every such $\kappa$, each term in the right-hand side of (\ref{eq:b2}) is $\lesssim n \epsilon_n^2$, and hence there exists a $c_3 > 0$ (depending on $\kappa$) such that
$$
  \log N(\mathcal{H}_n, \epsilon_n, H) \leq c_3 n \epsilon_n^2.
$$
This completes the verification of Condition B2 (for $\mathcal H_n$ determined by $\kappa$). It is also easy to show for these choices that
$$
  \kappa n \epsilon_n^2 \lesssim \min\left\{c_D d \log P, C_L L \log L - \log T, \frac{U^2}{2\sigma^2_\mu} - \log L - \log T, G - \log K_\gamma\right\}.
$$
Applying Lemma~\ref{lem:s5} and Lemma~\ref{lem:s6} we can find a sufficiently small constant $c' > 0$ (independent of $\kappa$) such that
$$
  \Pi(\theta \notin \mathcal{F}_n) \leq e^{-c' \kappa n \epsilon_n^2}.
$$
Taking $\kappa > \frac{c_2 + 4}{c'}$ and $c_4 = 1$, we have $\Pi[\theta \notin \mathcal{F}_n] \leq c_4 \exp\{-(c_2 + 4)n\epsilon_n^2\}$, which completes the verification of Condition B3. Thus, by verifying Conditions B1, B2, and B3, we have established the posterior contraction rate of $\epsilon_n$ in Theorem \ref{thm:posterior-concentration}, with $\epsilon_n = (\log n/n)^{\alpha/(2\alpha+D_0)}$.

\section{Proofs of Lemmas}
For the proofs of Lemma \ref{lem:s1} and \ref{lem:s3}, we refer to the proofs of Lemma 1 and 2 in \citet{orlandi2021density}, and for the proof of Lemma \ref{lem:s5}, we refer to the proof of Lemma 3 in \citet{linero2024generalized}. We provide proofs for Lemma \ref{lem:s2}, \ref{lem:s4}, \ref{lem:s6}, \ref{lem:s7}, and \ref{lem:s8} below.

\paragraph*{Proof of Lemma \ref{lem:s2}.} Let $\bar \varepsilon_n \propto (\log n / n)^{1/2}$ and let $C_{Gam} = K - 1$. Let $a_0 = \min_j \pi_\gamma(\gamma_{0j})$ where $\pi_\gamma(\gamma)$ denotes the prior mass function of the $\log \Gam(a_\gamma, b_\gamma)$ distribution. For large enough $n$ we have
\begin{align*}
  \Pi_\gamma(\|\gamma - \gamma_0\|_\infty \le \bar \varepsilon_n)
  &= \prod_{j = 1}^{K - 1} \Pi_\gamma(|\gamma_j - \gamma_{0j}| \le \bar \varepsilon_n)
  \\&\geq \prod_{j = 1}^{K - 1}\left\{ \bar \varepsilon_n \pi_\gamma(\gamma_{0j}) \right\}
  \\&\geq a_0^{K - 1} \bar \varepsilon_n^{K - 1}
  \\&\geq a_0^{K-1} \left( \frac{\log n}{n} \right)^{K-1}
  \\&\geq n^{-(K-1)} = e^{-(K-1) n \, \bar \varepsilon_n^2}.
\end{align*}
Here, the second line follows from the fact that $\pi_\gamma(\gamma)$ has positive density at $\gamma_{0j}$, while the final line follows from the fact that $(\log n)^{K-1} > a_0^{1 - K}$ eventually and that $n^{-1} = e^{-n \bar \varepsilon_n^2}$.

\paragraph*{Proof of Lemma \ref{lem:s4}.} The set $\mathcal{G}^\star_n$ is a $(K-1)$-dimensional hypercube with side length $2G$. To cover this with $\epsilon$-balls in the $L_\infty$-norm, we need no more than $\lceil 2G/2\epsilon \rceil = \lceil G/\epsilon \rceil$ evenly spaced intervals along each dimension. Thus the total number of balls required is no more than $\lceil G/\epsilon \rceil^{K-1} \lesssim \left(\frac{G}{\varepsilon}\right)^{K-1}$. This completes the proof.

\paragraph*{Proof of Lemma \ref{lem:s6}.}
We bound $\Pi(\gamma \notin \mathcal{G}^\star_n)$ as:
\begin{align*}
\Pi(\gamma \notin \mathcal{G}^\star_n) = \Pi(\|\gamma\|_\infty > G) & = \Pi(|\gamma_j| > G \text{ for some } j) \\
& \leq \sum_{j=1}^{K-1} \Pi(|\gamma_j| > G) \leq (K-1) \{\Pi(U > G) + \Pi(U < -G)\},
\end{align*}
where $U = \log V$ and $V \sim \text{Gamma}(a_\gamma, b_\gamma)$. We then bound the tail probabilities using Markov's inequality. Let $C_\gamma$ be such that $\E(V^{C_\gamma})$ and $\E(V^{-C_\gamma})$ both exist (i.e., $C_\gamma < a_\gamma$ for the appropriate moment of the inverse gamma to exist). Then
\begin{align*}
  \Pr(U > G) &= \Pr(V^{C_\gamma} > e^{C_\gamma \, G}) \le \E(V^{C_\gamma}) e^{-C_\gamma \, G}  \qquad \text{and} \\
  \Pr(U < -G) &= \Pr(V^{-C_\gamma} > e^{C_\gamma \, G}) \le \E(V^{-C_\gamma}) e^{-C_\gamma \, G}.
\end{align*}
Therefore, $\Pi(\gamma \notin \mathcal{G}^\star_n) \leq K_\gamma e^{-C_\gamma G}$, for $K_\gamma = (K-1) \{\E(V^{C_\gamma} + V^{-C_\gamma})\}$ and $C_\gamma < a_\gamma$. This completes the proof.

\paragraph*{Proof of Lemma \ref{lem:s7}.} We have, $Z_i = \mu + \epsilon_i$, with $\epsilon_i \sim h(\epsilon) = \exp(\epsilon - e^\epsilon)$, and $\mu = -r(x)$. Then the family $\{f(z \mid \mu): \mu \in \R \}$ is a location family, with location = $\mu$. We next verify the conditions of Corollary 1 (point 2) of \citet{linero2024generalized}.

Note that $h$ is twice continuously-differentiable function. Also, it satisfies Condition T of \citet{linero2024generalized} as:
\begin{enumerate}
  \item $h(x) \to 0$ as $x \to \pm \infty$,

  \item $S(\mu; Z) = \frac{\partial}{\partial \mu} \log h(Z - \mu) = e^{Z - \mu} - 1$ and it follows that $\E_{\mu = 0} \sup_{|\delta| \le 1} S(\delta; Z)^2 < \infty$.

  \item $I(\delta; Z) = -\frac{\partial^2}{\partial \mu^2} \log h(Z - \mu) = e^{Z - \mu}$ and it similarly follows that $\E_{\mu = 0} \sup_{|\delta| \le 1} I(\delta; Z) < \infty$.
\end{enumerate}
Thus, the family $\{f(z \mid \mu): \mu \in \R \}$ satisfies Condition T of \citet{linero2024generalized}.

\paragraph*{Proof of Lemma \ref{lem:s8}.}
By direct calculation,
\begin{align*}
  \Pr(Y > k \mid \gamma, r)
  &= \Pr(Z_1 > 0 \cap Z_2 > 0 \cap \cdots \cap Z_k > 0 \mid \gamma, r)
  \\&= \prod_{j = 1}^k \Pr(Z_k > 0 \mid \gamma, r)
  \\&= \prod_{j = 1}^k \exp(-e^{\gamma_k + r})
  \\&= \exp\left(-e^r \sum_{j = 1}^je^{\gamma_j}\right).
  \\&= \exp(-e^{c_k + r}).
\end{align*}


\end{document}